\documentclass[final]{siamltex}

\usepackage{amsmath}
\usepackage{amssymb}
\usepackage{mathrsfs}
\usepackage{graphicx}
\usepackage{enumerate}
\usepackage{psfrag}

\newcommand{\dps}[1]  {\displaystyle{#1} }

\def\R{\mathbb{R}}
\def\T{\mathbb{T}}
\def\Z{\mathbb{Z}}

\def\D{\mathcal{D}}
\def\A{\mathcal{A}}
\def\B{\mathcal{B}}
\def\L{\mathcal{L}}
\def\e{{\rm e}}
\newcommand{\ri}{\mathrm{i}}
\renewcommand{\ker}{\mathrm{Ker}}
\def\div{{\rm div}\, }

\newtheorem{rem}{Remark}
\newtheorem{theo}{Theorem}
\newtheorem{propo}{Proposition}

\newtheorem{lemm}{Lemma}
\newtheorem{corr}{Corollary}
\newtheorem{ass}{Assumption}

\title{Nonequilibrium shear viscosity computations with Langevin dynamics}

\author{
  R\'emi Joubaud
  \thanks{ANDRA, DRD/EAP, Parc de la croix blanche, 1,7 rue Jean Monnet, 92298 Ch\^{a}tenay-Malabry Cedex, France and Universit\'e Paris Est, CERMICS, Ecole des Ponts ParisTech, 6 \& 8 Av. Pascal, 77455 Marne-la-Vall\'ee, France.} 
  \and 
  Gabriel Stoltz
  \thanks{Universit\'e Paris Est, CERMICS and INRIA, MICMAC project-team, Ecole des Ponts ParisTech, 6 \& 8 Av. Pascal, 77455 Marne-la-Vall\'ee, France.}
}

\begin{document}

\maketitle

\begin{abstract}
We study the mathematical properties of a nonequilibrium Langevin dynamics which can be used to estimate the shear viscosity of a system. More precisely, we prove a linear response result which allows to relate averages over the nonequilibrium stationary state of the system to equilibrium canonical expectations. We then write a local conservation law for the average longitudinal velocity of the fluid, and show how, under some closure approximation, the viscosity can be extracted from this profile. We finally characterize the asymptotic behavior of the velocity profile, in the limit where either the transverse or the longitudinal friction go to infinity. Some numerical illustrations of the theoretical results are also presented.
\end{abstract}



\pagestyle{myheadings}
\thispagestyle{plain}
\markboth{R. Joubaud and G. Stoltz}{Nonequilibrium shear viscosity computations}

%-------------------------
%        INTRO
%------------------------

\section{Introduction}

The determination of the macroscopic properties of a material given its
microscopic description is the fundamental goal of statistical physics~\cite{Balian}.
Macroscopic properties can be classified into two categories: (i) equilibrium properties,
such as the heat capacity or the equation of state of the system 
(relating the pressure, the density and the temperature); and (ii) transport properties,
such as the thermal conductivity or the shear viscosity. 
The determination of transport properties is conceptually and numerically 
more challenging than the determination of equilibrium properties 
since transport phenomena depend both on the chosen thermodynamic ensemble 
and on the prescribed microscopic dynamics (which has to leave the thermodynamic state of 
the system invariant). Often, the Hamiltonian dynamics is considered as 
the reference dynamical evolution of the microscopic system. However, this dynamics
exactly preserves the energy of the system, while energy exchanges with the environment
are expected to happen. 
We believe that the choice of the underlying dynamics of the system is a modelling choice.
In any case, a careful study of the dependence of the computed transport properties
as a function of the parameters of the dynamics should be performed.

There are roughly three types of methods for computing transport coefficients:
\begin{enumerate}[(i)]
\item Transient methods, where the relaxation of some local disturbance
  is monitored as a function of time, and macroscopic coefficients are obtained
  by fitting the observed response to the evolution predicted by 
  a macroscopic evolution equation. For instance, 
  a local temperature hot spot can be created in the middle of a homogeneous material,
  and, assuming that the heat equation describes well the evolution of the kinetic temperature
  field, the diffusion of the energy allows to estimate the thermal conductivity
  of the material (see for instance~\cite{HRW05,stoltzPhD});
\item Equilibrium methods, which are based on time integrals of correlation functions,
  the so-called Green-Kubo formulas~\cite{Green,Kubo}. These correlation
  functions are obtained by sampling initial conditions according to the thermodynamic
  ensemble at hand, and averaging over all possible evolutions from these initial conditions.
  For example, the shear viscosity of a fluid is given by 
  \[
  \eta = \frac{|\mathcal{D}|}{k_{\rm B}T} 
  \int_0^{+\infty} \langle \sigma_{xy}(0)\sigma_{xy}(t)\rangle \, dt,
  \]
  where $|\mathcal{D}|$ is the volume of the simulation box, 
  $\sigma_{xy}(t)$ is the off-diagonal term of the microscopic stress tensor at time $t$,
  $T$ is the temperature of the system, and $\langle \cdot\rangle$ denotes a canonical 
  average;
\item Nonequilibrium methods, where the system is in a steady-state characterized by 
  the existence of stationary fluxes and spatial gradients of some quantities.
  These fluxes and gradients can be controlled by forcing terms acting on the boundaries
  of the system (for instance, a velocity profile can be obtained by 
  fixing the average velocity in the extremal slabs of a fluid, see the early 
  review~\cite{HooverAshurst}, or subsequent works such as~\cite{Ciccotti}), or by 
  a bulk process where fictitious forces act on all particles (the so-called synthetic 
  molecular dynamics approach~\cite{evans-moriss}). 
  Boundary driven methods are usually numerically less 
  efficient than bulk-driven methods since there are more correlations in the system
  and the convergence to a steady-state is slower.
\end{enumerate}
A review of the most standard approaches to compute the shear viscosity 
can be read in~\cite{evans-moriss}. See also~\cite{todd-daivis2007} for 
a focus on nonequilibrium methods. 

We focus in this study on bulk-driven nonequilibrium molecular dynamics techniques.
Since the system is driven out-of-equilibrium by a non-gradient force,
some thermostatting mechanism is required to prevent the uncontrolled
increase of the energy and to ensure that a steady-state can indeed be reached.
In many works focusing on the computation of shear viscosity, 
the thermostatting is performed with deterministic dynamics, 
such as Nos\'e-Hoover like thermostats~\cite{Nose84,Hoover85} 
or isokinetic dynamics (see~\cite[Chapter~3]{evans-moriss}). 
The ergodicity of these dynamics is at most
unclear (and non-ergodicity can even be proved rigorously in some limiting 
cases~\cite{LLM07,LLM09}).
The mathematical analysis of these methods is therefore untractable.
Only \emph{formal} results of linear response can be written down.
In some studies, the thermostatting is performed using 
dissipative particle dynamics~\cite{dpd1,dpd2}, which includes stochastic terms.
The ergodicity of these dynamics is however a very difficult issue, and the only existing results
we are aware of concern one-dimensional systems~\cite{SY06}.

We decided to use a standard Langevin dynamics
as the underlying dynamics of the system since this dynamics is ergodic and
has many nice mathematical properties, while still being close enough
to the Hamiltonian dynamics. 
In essence, the nonequilibrium dynamics we propose is obtained from the standard Langevin
dynamics by adding a nongradient force, which can be interpreted as some fictitious external
forcing term. The effect of this term is to create a velocity profile in the direction
of the forcing, and the viscosity of the fluid can be extracted from this profile.
The novelty of this work with respect to the (numerous) existing studies on the
computation of shear viscosity is the rigor of the mathematical arguments
used to prove linear response results and obtain the effective equation
on the observed velocity profile in terms of the applied external force. In particular,
we benefited from recent developments on hypocoercivity~\cite{Villani}. 
Besides, one of our main concern is the dependence of the viscosity as a function
of the parameters of the underlying dynamics, in particular the friction.
We analyzed the large friction asymptotics by extending and adapting 
mathematical studies of the auto-diffusion 
coefficient~\cite{HairerPavliotis04,pavliotis-stuart,HairerPavliotis08}. 
For the low friction dependence, we rely on numerical simulations.

\medskip

This paper is organized as follows.
We start by describing the Langevin dynamics we propose in Section~\ref{sec:dynamics},
and show the existence and uniqueness of the stationary state.
The mathematical properties of this stationary state are studied in Section~\ref{sec:math}
(with the proofs postponed to Section~\ref{sec:proofs}). In particular, we
give a rigorous proof of the linear response. We then show how to compute the 
viscosity, and characterize its asymptotic behavior for large frictions by determining
the limiting behavior of the velocity profile.
We finally present some numerical illustrations of the theoretical results
in Section~\ref{sec:num}.

%---------------------------------
%   Description de la dynamique
%---------------------------------

\section{The nonequilibrium Langevin dynamics}
\label{sec:dynamics}
\subsection{Description of the dynamics}

We consider a system of $N$ particles, enclosed in a periodic simulation box
$\D$. For simplicity of notation, we restrict ourselves to two-dimensional systems,
so that $\D = L_x\T \times L_y\T$ ($\T$ being the one-dimensional torus $\R/\Z$).
The particles are described by their positions $q=(q_1,\ldots,q_N) \in \D^N$
and their momenta $p=(p_1,\ldots,p_N) \in \mathbb{R}^{2N}$, and we assume that they have
identical masses $m >0$. Our results can however straightforwardly be extended
to the general case of particles with different masses and/or three-dimensional systems. 
We write $q_i=(q_{xi},q_{yi}) \in \D$, $p_i = (p_{xi},p_{yi}) \in \mathbb{R}^2$,
as well as $q_x=(q_{x1},\dots,q_{xN})$, $p_x=(p_{x1},\dots,p_{xN})$ and similar definitions for
$q_y,p_y$. The mass density of the system is 
\[
\rho=\frac{mN}{|\D|}, 
\]
$\left|\D\right|=L_xL_y$ being the volume of the box. Finally, we denote by  
\[
H(q,p) = \sum_{i=1}^N\frac{p_i^2}{2m} + V(q)
\]
the Hamiltonian of the system, the function $V$ being the potential energy.

The equations of motion we propose are a linear perturbation of the Langevin equations, with some
additional non-gradient external force in the $x$-direction (the direction of the flow). 
The dynamics reads (for $i=1,\dots,N$):
\begin{equation}
  \label{eq:langevinheq}
  \left \{ \begin{aligned}
    d q_{i,t} &= \frac{p_{i,t}}{m} \, dt,\\
    d p_{xi,t} &= -\nabla_{q_{xi}} V(q_t) \, dt + \xi F(q_{yi,t}) \, dt 
    - \gamma_x \frac{p_{xi,t}}{m} \, dt + \sqrt{\frac{2\gamma_x}{\beta}} \, dW^{xi}_t, \\
    d p_{yi,t} &= -\nabla_{q_{yi}} V(q_t) \, dt - \gamma_y \frac{p_{yi,t}}{m} \, dt 
    + \sqrt{\frac{2\gamma_y}{\beta}} \, dW^{yi}_t,
  \end{aligned} \right.
\end{equation}
where $\beta = 1/(k_{\rm B}T)$ is the inverse temperature, 
$\xi$ is the magnitude of the external nongradient force, 
$(W_t^x,W_t^y)_{t\ge0}$ is a $2N$-dimensional standard Brownian motion and 
the friction coefficients $\gamma_x,\gamma_y$ are real positive numbers.
In order to avoid irrelevant technical issues, we make the following assumption:

\begin{ass}
  \label{ass:smooth}
  The potential $V$  and the external force $F$ belong respectively
  to $C^\infty(\D^N)$ and $C^\infty(L_y\T)$.
\end{ass}

Note that the canonical measure 
\begin{equation}
  \label{eq:eq_inv_measure}
  \psi_{0}(q,p) \, dq \, dp =  Z^{-1} \, \e^{-\beta H(q,p)} \, dq \, dp
\end{equation}
is invariant by the dynamics~\eqref{eq:langevinheq} when $\xi = 0$,
and is actually the only invariant measure
(see for instance the discussion in~\cite{stoltz}).

\subsection{Existence and uniqueness of an invariant measure}

When $\xi \neq 0$, there is no obvious invariant probability measure, and the very existence
of such a measure is not guaranteed \textit{a priori}.
However, in the case when $\gamma_x,\gamma_y>0$,
standard techniques based on Lyapunov functions and hypoellipticity arguments can be resorted 
to to prove the existence and uniqueness of an invariant measure
which has a smooth density with respect to the Lebesgue measure.
For further purposes, it is convenient to consider the reference space $L^2(\psi_{0})$
(where the measure~$\psi_0$ is defined~\eqref{eq:eq_inv_measure}), 
endowed with the scalar product 
\[
\left\langle f,g\right\rangle_{L^2(\psi_0)} := \int_{\mathcal{D}^N\times \R^{2N}} f(q,p) g(q,p)
\, \psi_{0}(q,p) \, dq \, dp.
\]
We can then state the following result.

\begin{theo}
  \label{theo1}
  Consider $\gamma_x,\gamma_y>0$ and suppose that Assumption~\ref{ass:smooth} holds. 
  Then, for any $\xi \in \mathbb{R}$, 
  the dynamics~\eqref{eq:langevinheq} has a unique smooth invariant measure 
  with density $\psi_{\xi}\in C^{\infty}(\mathcal{D}^N\times \R^{2N})$. 
  Besides, there exists $\xi^* > 0$ such that,
  for any $\xi \in (-\xi^*,\xi^*)$, the following expansion 
  holds in $L^2(\psi_0)$:
  \begin{equation}
    \label{eq:expansion_psi_xi}
    \psi_{\xi} = f_{\xi}\psi_{0}, \qquad f_{\xi} = 1 + \sum_{k\ge 1} \xi^k \mathrm{f}_k,
  \end{equation}
  where $\mathrm{f}_k \in L^2(\psi_0)$ is such that 
  $\| \mathrm{f}_k \|_{L^2(\psi_0)} \leq C (\xi^*)^{-k}$.
\end{theo}

The proof is presented in Section~\ref{sec:proof_psi_xi}.
The existence and uniqueness of an invariant measure in the case when either
$\gamma_x = 0$ or $\gamma_y = 0$ is a much more difficult question.
To obtain such a result, more precise assumptions on the potential are required
(see Remark~\ref{rmk:gamma=0} in Section~\ref{sec:proof_psi_xi}). 

In the sequel, averages with respect to the measure $\psi_\xi$ are denoted
\[
\langle h \rangle_\xi = \int_{\D^N \times \R^{2N}} h(q,p) \, \psi_\xi(q,p) \, dq \, dp = 
\langle h, f_\xi\rangle_{L^2(\psi_0)}.
\]

%-------------------------
%   Analyse mathematique
%-------------------------

\section{Mathematical analysis of the viscosity}
\label{sec:math}

\subsection{Linear response} 
\label{sec:linear_response}

Linear response results allow to compute the average of some properties with respect to 
the nonequilibrium measure in terms of equilibrium averages, in the limit
when the parameter giving the strength of the nonequilibrium forcing vanishes.
To describe the result more precisely, we introduce 
the infinitesimal generator associated to the equilibrium Langevin process 
(\emph{i.e.}~\eqref{eq:langevinheq} in the case when $\xi=0$): 
\[
\A_{0} = \A_{\rm ham} + \A_{\rm thm}, 
\]
where
\[
\A_{\rm ham} = \frac{p}{m} \cdot \nabla_q - \nabla V(q) \cdot \nabla_p, 
\]
and
\[
\A_{\rm thm} = \sum_{\alpha=x,y} \gamma_\alpha \left(-\frac{p_\alpha}{m}\cdot \nabla_{p_\alpha} 
+ \frac1\beta \Delta_{p_\alpha} \right) 
= \frac{\e^{\beta H}}{\beta} 
\sum_{\alpha=x,y} \gamma_\alpha 
\div_{{p_\alpha}} \left(\e^{-\beta H}\nabla_{p_\alpha} \cdot\right).
\]
It can be proved (see Section~\ref{sec:proof_psi_xi}) 
that $\A_0^{-1}$ is a well defined operator on the Hilbert space
\[
\mathcal{H} = \left\{ f \in L^2(\psi_0) \, \left| \, \int_{\D^N \times \mathbb{R}^{2N}} 
f \psi_0 = 0 \right.\right \} 
= L^2(\psi_0) \cap \{ 1 \}^\perp,
\]
where the orthogonality is with respect to the $L^2(\psi_0)$ scalar product.

Define also the adjoint operator on $L^2(\psi_{0})$ of the generator:
\[
\L_0 = \A_0^* = - \A_{\rm ham} + \A_{\rm thm}.
\]
The generator of the nonequilibrium perturbation reads
\[
\B = \sum_{i=1}^N F(q_{yi}) \partial_{p_{xi}},
\]
and its adjoint on $L^2(\psi_{0})$ is
\[
\B^* = - \sum_{i=1}^N F(q_{yi}) \partial_{p_{xi}} + \frac{\beta}{m} p_{xi} F(q_{yi}).
\]
The generator of the dynamics~\eqref{eq:langevinheq} is therefore 
\[
\A_{\xi} = \A_{0}+ \xi \B,
\]
with adjoint $\L_{\xi} = \L_0 + \xi \B^*$.

Linear response is an easy consequence of Theorem~\ref{theo1}:

\begin{corr}
  \label{replin}
  Under the same assumptions as in Theorem~\ref{theo1}, and for any function $h$ regular enough,
  \begin{equation}
    \label{eq:lin}
    \lim_{\xi \rightarrow 0} \frac{\left\langle \A_0 h \right\rangle_\xi}{\xi} 
    = - \frac{\beta}{m} 
    \left\langle h, \sum_{i=1}^N p_{xi} F(q_{yi}) \right\rangle_{L^2(\psi_0)}.
  \end{equation}
  Besides, for any function $h \in \mathcal{H}$,
  \[
  \lim_{\xi \rightarrow 0} \frac{\left\langle h \right\rangle_\xi}{\xi} 
    = - \frac{\beta}{m} 
    \left\langle \A_0^{-1} h, \sum_{i=1}^N p_{xi} F(q_{yi}) \right\rangle_{L^2(\psi_0)}.
  \]
\end{corr}
\begin{proof}
Since $\langle \A_0 h, 1 \rangle_{L^2(\psi_0)} = \langle h, \L_0 1 \rangle_{L^2(\psi_0)} = 0$,
it holds
\[
\left\langle \A_0 h \right\rangle_\xi = \langle \A_0 h,f_\xi \rangle_{L^2(\psi_0)}
= \xi \langle \A_0 h,\mathrm{f}_1 \rangle_{L^2(\psi_0)} + \mathrm{O}(\xi^2)
= \xi \langle h, \L_0\mathrm{f}_1 \rangle_{L^2(\psi_0)} + \mathrm{O}(\xi^2).
\]
Now, in the proof of Theorem~\ref{theo1}, we show that (see~\eqref{eq:L0f1})
\[
\L_0\mathrm{f}_1 = - \B^* 1 = -\frac{\beta}{m}\sum_{i=1}^N p_{xi} F(q_{yi}),
\]
which gives the expected result. \qquad
\end{proof}

\subsection{Local conservation of the longitudinal velocity}

We prove in this section a conservation equation for 
velocities in the $x$-direction, when spatial averages over small windows 
in the transverse direction~$y$ are considered. This allows to state
an equation relating the off-diagonal term of the stress tensor and
the nongradient force acting on the system, see~\eqref{eq:local_conservation} below.
Our derivation may be seen as a mathematically rigorous counterpart to the 
seminal work of Irving and Kirkwood~\cite{irving-kirkwood}.
We assume from now on that the potential energy is given by a sum of pairwise interactions: 
\begin{equation}
  \label{eq:potential_pairwise}
  V(q_1,\ldots,q_N)=\sum_{1\le i < j\le N} \mathcal{V}(|q_i-q_j|),
\end{equation}
for some given smooth potential~$\mathcal{V}$.

Consider the following average longitudinal velocity:
\begin{equation}
  \label{eq:def_U}
  U_x^\varepsilon(Y,q,p) = \frac{L_y}{Nm}\sum_{i=1}^N p_{xi}
  \chi_{\varepsilon}\left(q_{yi}-Y\right),
\end{equation}
where $\chi_{\varepsilon}$ (with $0 <\varepsilon \leq 1$) is an approximation of the identity 
on $L_y\T$. 
More precisely, 
\[
\chi_{\varepsilon}(s) = \sum_{n \in \mathbb{Z}} 
\frac{1}{\varepsilon}\chi\left(\frac{s-n L_y}{\varepsilon}\right),
\]
where $\chi\in C^{\infty}(\R)$ has support in~$[0,L_y]$ and $\int_0^{L_y} \chi = 1$.
The factor $L_y$ in~\eqref{eq:def_U} accounts for the fact that $\chi_\varepsilon$
has units of inverse lengths: in fact,
\[
\frac{1}{L_y}\int_0^{L_y} U_x^\varepsilon(Y,q,p) \, dY = \frac{1}{Nm}\sum_{i=1}^N p_{xi}
\]
is the average velocity of the system.
In practice, averages such as~\eqref{eq:def_U} are computed with bin indicator functions (see 
Section~\ref{sec:numerical_implementation}).

We also need a spatially localized (with respect to the altitude~$Y$)
version of the off-diagonal term of the stress tensor.
This quantity is given by the following expression (the fact that it can be interpreted
as some stress tensor is motivated below by the limiting spatial 
average~\eqref{eq:spatial_avg_sigma} as well
as the conservation law~\eqref{eq:local_conservation}):
\begin{equation}
  \label{eq:Sigma}
  \begin{aligned}
  & \Sigma_{xy}^\varepsilon(Y,q,p) \\
  & = \frac{1}{L_x} \left( \sum_{i=1}^N \frac{p_{xi} p_{yi}}{m}\chi_{\varepsilon}\left(q_{yi}-Y\right)
  - \! \! \!
  \sum_{1 \leq i < j \leq N} \! \! \! 
  \mathcal{V}'(|q_i-q_j|)\frac{ q_{xi}-q_{xj}}{|q_i-q_j|}
  \int_{q_{yj}}^{q_{yi}} \chi_{\varepsilon}(s-Y) \, ds \right).
  \end{aligned}
\end{equation}
Note that the limiting spatial average over $Y$ 
\begin{equation}
  \label{eq:spatial_avg_sigma}
  \begin{aligned}
    & \lim_{\varepsilon \to 0} \frac{1}{L_y} \int_0^{L_y} \Sigma_{xy}^\varepsilon(Y,q,p) \, dY \\ 
    & \qquad =\frac{1}{L_xL_y} \left( \sum_{i=1}^N \frac{p_{xi} p_{yi}}{m} 
    -\sum_{1 \leq i < j \leq N}
    \mathcal{V}'(|q_i-q_j|)\frac{ (q_{xi}-q_{xj})(q_{yi}-q_{yj})}{|q_i-q_j|}
    \right)
    \end{aligned}
\end{equation}
is the standard expression encountered for the off-diagonal term 
of the pressure tensor without spatial localization. The expression~\eqref{eq:Sigma}
comes out naturally from the mathematical analysis 
(see the proof of Proposition~\ref{prop:local_conservation}), 
and was already proposed in~\cite{todd-pressure} (where it is called the 'method of planes'). 

The relationship between the local longitudinal velocity and the off-diagonal term of the 
stress tensor is made precise in the following proposition.

\begin{propo}
\label{prop:local_conservation}
The limits
\[
u_x(Y) = \lim_{\varepsilon \to 0}
\lim_{\xi \to 0} \frac{\left\langle U_x^\varepsilon(Y,\cdot)\right\rangle_\xi}{\xi}
\]
and 
\[
\sigma_{xy}(Y) = \lim_{\varepsilon \to 0}
\lim_{\xi \to 0} \frac{\left\langle \Sigma_{xy}^\varepsilon(Y,\cdot)\right\rangle_\xi}{\xi}
\]
belong to $C^\infty(L_y\mathbb{T})$ and 
\begin{equation}
  \label{eq:local_conservation}
  \frac{d\sigma_{xy}(Y)}{dY} + \gamma_{x} \overline{\rho} u_x(Y) = \overline{\rho} F(Y),
\end{equation}
where $\overline{\rho} = \rho/m$ is the particle density. 
\end{propo}

The proof is based on an application of Corollary~\ref{replin} with~\eqref{eq:def_U} as a test
function~$h$. The order of the limits $\varepsilon \to 0$ and
$\xi \to 0$ cannot be inverted since the linear response result
of Corollary~\ref{replin} cannot be applied with $h$ replaced by the limit 
of $\chi_\varepsilon$ (which is a Dirac mass).

Equations similar to~\eqref{eq:local_conservation} could be written down for other quantities 
such as the transverse velocity~$U_y$ and 
longitudinal and transverse energy fluxes 
(see~\cite{irving-kirkwood} for the original derivation of the corresponding equations).

%----------- closure -----------------------
\subsection{Definition and closure relation for shear viscosity computations}

We now discuss a closure relation for~\eqref{eq:local_conservation}, which allows
to obtain an equation on the average velocity only,
from which the viscosity can be extracted.

By analogy with continuum fluid mechanics, we \emph{define} the shear viscosity $\eta$ as
follows:
\begin{equation}
  \label{eq:Newton_law}
  \sigma_{xy}(Y) := -\eta(Y)\dfrac{du_x(Y)}{dY}. 
\end{equation}
This definition leads to the following equation on $u_x$: 
\[
-\frac{d}{dY}\left(\eta(Y) \, \frac{du_x(Y)}{dY}\right) + \gamma_x\overline{\rho}u_x(Y) = 
\overline{\rho} F(Y).
\]
In bulk homogeneous fluids, the simplest closure is to assume that
\begin{equation}
  \label{eq:eta_constant}
  \eta(Y) = \eta > 0,
\end{equation}
so that the following equation on $u_x$ is obtained:
\begin{equation}
  \label{eq:stokesfriction}
  -\eta u_x''(Y) + \gamma_x \overline{\rho} u_x(Y) = \overline{\rho} F(Y).
\end{equation}
In order to ensure the uniqueness of the solution when $\gamma_x = 0$,
an additional condition should be added (such as a vanishing integral over
the domain $L_y \mathbb{T}$).

The equation~\eqref{eq:stokesfriction} obtained with the help of the
closure relation is the basis for numerical methods to compute the 
shear viscosity given a potential energy function~$V$.
We were not able to justify mathematically the assumption~\eqref{eq:eta_constant}.
We nonetheless provide a numerical validation of this assumption in 
Section~\ref{sec:validation_closure}.

%--------- asymptotiques ----------------
\subsection{Asymptotic behaviour of the viscosity for large frictions}
\label{sec:large_gamma}

An important issue is the dependence of the viscosity on the parameters
of the dynamics. For the Langevin dynamics~\eqref{eq:langevinheq},
this means understanding the dependence of the viscosity on the friction 
parameters $\gamma_x,\gamma_y$. The limits $\gamma_x \to 0$ 
or $\gamma_y \to 0$ are very difficult to study mathematically
without strong assumptions on the potential and/or the geometry
of the system (see Remark~\ref{rmk:gamma=0}). We therefore rely on numerical
simulations for these cases (see Sections~\ref{sec:limite_num_gamma_y} 
and~\ref{sec:limite_num_gamma_x}).

On the other hand, the limit when one of the friction parameters goes to 
infinity can be studied. To this end, we have to 
understand the limit of the velocity field~$u_x$ as either
$\gamma_x$ or $\gamma_y$ goes to infinity. This is done by rigorous asymptotic analysis.
Thanks to~\eqref{eq:stokesfriction}, limiting behaviors of the viscosity may
be inferred from the limiting behaviors of the velocity profiles.
The key result to obtain the limiting velocity profile is to characterize
the limit of some averages with respect to specific solutions of the 
Poisson equation (see~\eqref{eq:eq_Poisson_y} and~\eqref{eq:eq_Poisson_x} below).

\subsubsection{Infinite transverse friction}

We start with the case $\gamma_y \to +\infty$, for a fixed value $\gamma_x > 0$.

\begin{theo}[Infinite transverse friction]
\label{prop:gamma_y}
Consider a given smooth function $G$ and a longitudinal friction $\gamma_x > 0$. 
Define $\A_{0}(\gamma_y):=\A_0=\A_{\rm ham} + \gamma_x \A_{x, \rm thm} 
+ \gamma_y \A_{y, \rm thm}$, with
\begin{equation}
  \label{eq:def_A_alpha}
  \A_{\alpha, \rm thm} = -\frac{p_\alpha}{m}\cdot \nabla_{p_\alpha} + \frac1\beta \Delta_{p_\alpha},
\end{equation}
and denote by $f_{\gamma_y}$ the unique solution in~$\mathcal{H}$ of the equation
\begin{equation}
  \label{eq:eq_Poisson_y}
  -\A_{0}(\gamma_y) f_{\gamma_y} = \sum_{i=1}^N p_{xi} G(q_{yi}).
\end{equation}
Then, there exist $f^0,f^1 \in H^1(\psi_0)$ 
and a constant $C>0$ such that, for all $\gamma_{y}\ge \gamma_{x}$,
\begin{equation}
  \label{eq:cv_Poisson_gamma_y}
\left\| f_{\gamma_{y}}- f^0-\gamma_{y}^{-1}f^1\right\|_{H^1(\psi_0)} \le \frac{C}{\gamma_y}.
\end{equation}
Besides, the function $f^0$ is of the general form 
\[
f^0(q,p) = \sum_{i=1}^N G (q_{yi}) \phi_i(q_x,q_y,p_x),
\]
where the functions $\phi_i$ are $C^\infty$.
\end{theo}

The proof can be read in Section~\ref{sec:proof_gamma_y}.
The above result can be used to understand the limit of $u_x(Y)$ as $\gamma_y \to +\infty$.
Indeed, by Proposition~\ref{replin},
\[
u_x^{\gamma_y,\varepsilon}(Y) := 
\lim_{\xi \to 0} \frac{\left\langle U_x^\varepsilon(Y,\cdot)\right\rangle_\xi}{\xi}
= -\frac{\beta}{m} \left \langle \sum_{i=1}^N p_{xi} F(q_{yi}), 
\mathscr{U}_{\gamma_y}^\varepsilon(Y,q,p) \right \rangle_{L^2(\psi_0)},
\]
where $-\A_{0}(\gamma_y) \mathscr{U}_{\gamma_y}^\varepsilon(Y,\cdot) = U_x^\varepsilon(Y,\cdot)$
is a Poisson equation of the form~\eqref{eq:eq_Poisson_y}
(with $G(y)$ proportional to $ \chi_\varepsilon(y-Y)$).
The convergence result~\eqref{eq:cv_Poisson_gamma_y} shows that 
$\mathscr{U}_{\gamma_y}^\varepsilon(Y,\cdot)$ has a limit as $\gamma_y \to +\infty$, and  
the limiting velocity field reads
\[
u_x^{\infty,\varepsilon}(Y) = \frac{\beta L_y}{Nm^2} \left \langle \sum_{i=1}^N p_{xi} F(q_{yi}), 
\sum_{j=1}^N \chi_\varepsilon(q_{yj}-Y) \phi_j(q_x,q_y,p_x) \right \rangle_{L^2(\psi_0)}.
\]
The latter quantity has a limit as $\varepsilon \to 0$, so 
that the velocity field converges to some limiting field~$u_x^{\infty}$. 
Therefore, the viscosity extracted
from~\eqref{eq:stokesfriction} also has a finite limit.
These theoretical considerations are illustrated by 
numerical simulations in Section~\ref{sec:limite_num_gamma_y}.

\subsubsection{Infinite longitudinal friction}
\label{sec:infinite_gamma_x}

We now consider the limit $\gamma_x \to +\infty$, for a fixed value $\gamma_y > 0$.
In this case, the leading term of the expansion in inverse powers of $\gamma_x$
is~0, and a refined convergence result is needed to discuss the limit of the velocity
profile.

\begin{theo}[Infinite longitudinal friction]
\label{prop:gamma_x}
Consider a given smooth function $G$ and a transverse friction $\gamma_y > 0$. 
Define $\A_{0}(\gamma_x):=\A_0=\A_{\rm ham} + \gamma_x \A_{x, \rm thm} 
+ \gamma_y \A_{y, \rm thm}$,
and denote by $f_{\gamma_x}$ the unique solution in~$\mathcal{H}$ of the equation
\begin{equation}
  \label{eq:eq_Poisson_x}
  -\A_{0}(\gamma_x) f_{\gamma_x} = \sum_{i=1}^N p_{xi} G(q_{yi}).
\end{equation}
Then, there exist $f^1,f^2 \in H^1(\psi_0)$ 
and a constant $C>0$ such that, for all $\gamma_{x}\ge \gamma_{y}$,
\begin{equation}
  \label{eq:cv_Poisson_gamma_x}
\left\| f_{\gamma_{x}}-\gamma_{x}^{-1}f^1-\gamma_{x}^{-2}f^2\right\|_{H^1(\psi_0)} 
\le \frac{C}{\gamma^2_x}.
\end{equation}
Besides, the dependence of the function $f^1$ in the variable $p_x$ can be 
written explicitly as
\[
f^1(q,p) = m \sum_{i=1}^N p_{xi}G(q_{yi}) + \widetilde{f}^1(q,p_y).
\]
\end{theo}

The proof can be read in Section~\ref{sec:proof_gamma_x}. 
To obtain asymptotics on the velocity field, we apply the above convergence result
with $G(y)$ proportional to $ \chi_\varepsilon(y-Y)$ (denoting by $f^1_\varepsilon$ the 
first term in the expansion in inverse powers of~$\gamma_x$):
\[
\begin{aligned}
u_x^\varepsilon(Y) := 
\lim_{\xi \to 0} \frac{\left\langle U_x^\varepsilon(Y,\cdot)\right\rangle_\xi}{\xi}
& = \frac{\beta L_y}{N m^2 \gamma_x} \left \langle \sum_{i=1}^N p_{xi} F(q_{yi}), 
f^1_\varepsilon(q,p) \right \rangle_{L^2(\psi_0)} + \mathrm{O}\left(\frac{1}{\gamma_x^2}\right) \\
& = \frac{\beta L_y}{N m \gamma_x} \left \langle \sum_{i=1}^N p_{xi} F(q_{yi}), 
\sum_{j=1}^N p_{xj}\chi_\varepsilon(q_{yj}-Y) \right \rangle_{L^2(\psi_0)} 
\! \! \! \! \! \! \! \! \! \! + \mathrm{O}\left(\frac{1}{\gamma_x^2}\right) \\
& = \frac{L_y}{\gamma_x} \int_0^{L_y} F(y) \chi_\varepsilon(y-Y) \, dy 
+ \mathrm{O}\left(\frac{1}{\gamma_x^2}\right),
\end{aligned}
\]
where we have used the fact that $\langle p_{xi}, \widetilde{f}_\varepsilon^1 
\rangle_{L^2(\psi_0)} = 0$ since $\widetilde{f}_\varepsilon^1$ does not depend on~$p_x$.
This shows that the following limit is well defined:
\begin{equation}
  \label{eq:limit_velocity_gamma_x}
  \overline{u}_x(Y) = \lim_{\varepsilon \to 0} \lim_{\gamma_x \to +\infty} 
  \gamma_x u_x^\varepsilon(Y) = F(Y).
\end{equation}
The limiting velocity profile $\overline{u}_x$ 
does not depend on the specific interaction potential~$\mathcal{V}$,
and is the same for all systems with pairwise interactions.
Besides, the viscosity~$\eta$ cannot be extracted from~\eqref{eq:stokesfriction}
since $(u_x^\varepsilon)''$ is of order $\gamma_x^{-1}$ while $F$ and $\gamma^\varepsilon_x u_x$ 
are of order~1.
The limit $\gamma_x \to +\infty$ is therefore somewhat degenerate from a theoretical viewpoint.
Numerical simulations however allow to investigate the large $\gamma_x$ asymptotics,
see Section~\ref{sec:limite_num_gamma_x}.

%--------------------------
%       NUMERIQUE 
%--------------------------

\section{Numerical results for the Lennard-Jones fluid}
\label{sec:num}

We present in this section some numerical illustrations of the theoretical
results obtained in Section~\ref{sec:math}.

%----------- estimateurs ---------------
\subsection{Numerical implementation}
\label{sec:numerical_implementation}

\paragraph{Description of the system}
We consider a Lennard-Jones fluid, which is a standard test case for shear flow
computations, in a 2-dimensional setting (in order
to limit the number of degrees of freedom and henceforth obtain 
results with lower statistical uncertainties).
The potential energy is of the form~\eqref{eq:potential_pairwise}, with
\begin{equation}
\mathcal{V}_{\rm LJ}(r) = 4\varepsilon_{\rm LJ} 
\left( \left(\frac{d_{\rm LJ}}{r}\right)^{12}-\left(\frac{d_{\rm LJ}}{r}\right)^{6} \right).
\end{equation}
Actually, it is numerically more convenient to work with a truncated potential, which reads:
\[
\mathcal{V}(r) = \begin{cases}
\mathcal{V}_{\rm LJ}(r) & \quad \mathrm{if} \ r \leq r_{\rm spline},\\
v_{\rm spline}(r) & \quad \mathrm{if} \ r_{\rm spline} \leq r \leq r_{\rm cut},\\
0 & \quad \mathrm{if} \ r \geq r_{\rm cut}.
\end{cases}
\]
The function $v_{\rm spline}$ is a polynomial of order~3 which is such that the potential
is~$C^1$ on $(0,+\infty)$
(note that there is a singularity at $r=0$ so that this potential does not 
satisfy Assumption~\ref{ass:smooth}. However, it seems that this singularity does not
show up in the numerical simulations. Any problem related to this singularity
could be overcome by modifying appropriately the potential for the very small values of~$r$).
We use $r_{\rm cut}=3d_{\rm LJ}$ and $r_{\rm spline}=0.9\,r_{\rm cut}$. 

All the results presented below are in reduced units, which are 
determined by setting to~1 
the energy~$\varepsilon_{\rm LJ}$, the length~$d_{\rm LJ}$, and the mass~$m$.
The remaining tunable parameters of the model are the force amplitude~$\xi$ and the friction 
parameters~$\gamma_x,\, \gamma_y$.

The thermodynamic state of the system is determined by 
the fluid mass density $\rho$ and the temperature $T$. In the numerical illustrations
below, we set $\beta = 0.4$, 
$\rho = 0.69$ and consider $L_x = 360$ and $L_y = 18$. The number of simulated
particles is therefore $N =4500$.

We have checked that the thermodynamic limit is attained for the systems we simulate,
\emph{i.e.} that the values of the viscosity and the profiles we present 
are converged with respect to increasing values of $L_x,L_y$ (at fixed density), see~\cite{Remi}.

\paragraph{Nongradient forces}
We consider three different external perturbations, which are all normalized
so that $-1 \leq F(y) \leq 1$:
\begin{enumerate}[(i)]
\item sinusoidal perturbation:
  $\dps F(y)=\sin\left(\frac{2\pi y}{L_y}\right)$;
\item piecewise linear perturbation:
  $\dps F(y)= \begin{cases}
  \displaystyle{ \frac{4}{L_y}\left(y-\frac{L_y}{4}\right) }, 
  & \displaystyle{0 \leq y \le \frac{L_y}{2}},\\[10pt]
  \displaystyle{\frac{4}{L_y}\left(\frac{3L_y}{4}-y\right)}, 
  & \displaystyle{\frac{L_y}{2} \leq y \leq L_y};\\
\end{cases}
$
\item piecewise constant constant perturbation:
  $\dps 
  F(y)= \begin{cases}
  \dps 1, & \displaystyle{0 < y < \frac{L_y}{2}},\\[10pt]
  \dps -1, & \displaystyle{\frac{L_y}{2}< y < L_y}.\\
  \end{cases}$
\end{enumerate}
Note that only the sinusoidal force satisfies Assumption~\ref{ass:smooth}.
This shape of perturbation, introduced in~\cite{gosling73}, 
is the most popular choice for shear viscosity computations.

\paragraph{Integration of the dynamics}
The dynamics~\eqref{eq:langevinheq} is discretized using a standard splitting scheme,
similar to the schemes proposed in~\cite{LRS10}. 
The evolution is decomposed as the superposition of (i) a Hamiltonian part, which is
integrated with the standard Verlet scheme~\cite{Verlet,HairerLubichWanner06}; 
and (ii) a fluctuation/dissipation part containing 
also the nongradient force, which can be integrated analytically since it is
an Ornstein-Uhlenbeck process with a constant drift.
The numerical scheme reads
\begin{equation}
\label{eq:numerical_scheme}
\left\{
\begin{aligned}
  p^{n+1/4} &= p^n - \frac{\Delta t}{2} \nabla V(q^n),\\
  q^{n+1}   &= q^n + \Delta t \, p^{n+1/4}, \\
  p^{n+1/2} &= p^{n+1/4} - \frac{\Delta t}{2} \nabla V(q^{n+1}),\\
  p_{xi}^{n+1} &= \alpha_x p_{xi}^{n+1/2} + \sqrt{\frac{1}{\beta}(1-\alpha_x^2)} \, G_{xi}^n 
  + \left(1-\alpha_x\right) \frac{\xi}{\gamma_x}F\left(q_{yi}^{n+1}\right), \quad i=1,\dots,N\\
  p_y^{n+1} &= \alpha_y p_y^{n+1/2} + \sqrt{\frac{1}{\beta}(1-\alpha_y^2)} G_y^n,
\end{aligned}
\right.
\end{equation}
where $\alpha_{x,y} = \exp(-\gamma_{x,y}\Delta t)$, and $G_x^n,G_y^n$ 
are independent and identically distributed standard Gaussian random variables.
Note that this scheme is well behaved in the limits $\gamma_x \to +\infty$ and/or
$\gamma_y \to +\infty$, as well as in the limits $\gamma_y \rightarrow 0$ or
$\gamma_x \rightarrow 0$ (the well posedness of the latter case
is a consequence of the limit 
$\left(1-\alpha_x\right)/\gamma_x \rightarrow \Delta t$ as $\gamma_x \to 0$).
It reduces to the standard Verlet scheme when $F=0$ and $\gamma_x = \gamma_y = 0$.

We use $\Delta t=0.005$ in all the simulations below. This time step ensures that
the relative error in energy is about 1\% for the Verlet scheme. 

\paragraph{Numerical localization}
To analyze the various fields which can be constructed from the numerical data
generated by the simulation (longitudinal velocity, off-diagonal component of the stress
tensor, kinetic temperatures, etc), we use a binning procedure in the $Y$ variable.
More precisely, we introduce a mesh with a uniform spacing $\Delta Y$, centered
on the altitudes $Y_s = (s+1/2) \Delta Y$ (with $0 \leq s \leq S-1$ and $S \Delta Y = L_y$).

The microscopic observables we wish to average are either the longitudinal 
velocity~\eqref{eq:def_U} or the off diagonal stress tensor~\eqref{eq:Sigma}.
Both functions are of the general form
\[
A^\varepsilon(Y,q,p) = \sum_{i=1}^N a_i(q,p) \Phi_\varepsilon(q_{yi}-Y),
\]
where $\Phi_\varepsilon$ is either $\chi_\varepsilon$ or an integral of this function.
Averages of such functions over each cell are computed as 
\[
\mathscr{A}^\varepsilon_s = \frac{1}{\xi \Delta Y} 
\int_{Y_s - \Delta Y/2}^{Y_s + \Delta Y/2} \left \langle A(Y,\cdot) \right
\rangle_{\xi} \, dY = \frac{1}{\xi \Delta Y}  
\left \langle \int_{Y_s - \Delta Y/2}^{Y_s + \Delta Y/2} 
A^\varepsilon(Y,\cdot) \, dY \right
\rangle_{\xi}.
\]
Taking advantage of the integration in the $Y$ variable, it is possible
to take the limit $\varepsilon \to 0$ in the latter expression. This amounts to compute
ensemble averages with respect to bin indicator functions (or their integrals).
For instance, the average longitudinal velocity in the $s$th bin is
\[
\mathscr{U}_s = \frac{L_y}{\xi N m \Delta Y} \left \langle \sum_{i=1}^N p_{xi} 
\mathbf{1}_{[Y_s - \Delta Y/2,Y_s + \Delta Y/2]}(q_{yi}) 
\right \rangle_{\xi}.
\]
In practice, the ensemble average $\langle \cdot \rangle_{\xi}$ 
is computed as a time average over trajectories
$(q^n,p^n)_{n=1,\dots,N_{\rm iter}}$.

\paragraph{Estimation of the viscosity}
The solutions of Equation~\eqref{eq:stokesfriction} are periodic in the $Y$-variable and 
are hence most easily
analyzed using Fourier series (see for instance the discussion 
in~\cite{evans-moriss,todd-hansen2007}).
We consider the Fourier coefficients 
of the the average longitudinal velocity $u_x$ and the force~$F$, 
given respectively for $k\in \Z$ by
\begin{equation}
  \label{eq:fourier}
  U_k = \frac{1}{L_y}\int_0^{L_y} u_x(y)\exp\left(\frac{2\ri k\pi y}{L_y}\right)dy,
  \quad
  F_k = \frac{1}{L_y}\int_0^{L_y} F(y)\exp\left(\frac{2\ri k\pi y}{L_y}\right)dy.
\end{equation}
The coefficients~$U_k$ can be esimtated numerically using trajectory averages as
\begin{equation}
  \label{eq:fourier1}
  U_k^{N_{\rm iter}}= 
  \frac{1}{N_{\rm iter} \xi N} \sum_{n=1}^{N_{\rm iter}} \sum_{j=1}^N \frac{p_{xj}^n}{m}
  \exp\left(\frac{2\ri k\pi q_{yj}^m}{L_y}\right).
\end{equation}
This is a valid estimation provided the marginal distribution in the position 
of one particle is the uniform law on the domain. By translation invariance, 
this is true when no external force is present. It remains approximately
true when $\xi$ is not too large. We checked that this approximation has no influence
on the presented numerical results.

The shear viscosity is obtained from a Fourier analysis of~\eqref{eq:stokesfriction}.
The value of $\eta$ should be independent of $k \in \mathbb{Z}$. It should also satisfy
the following equation: 
\begin{equation}
  \label{eq:fourierU}
 U_k = \frac{F_k}{\dfrac{\eta}{\overline{\rho}}\left(\dfrac{2\pi}{L_y}\right)^2 k^2+\gamma_x}.
\end{equation}
The shear viscosity is finally obtained as
\begin{equation}
  \label{eq:viscosity}
  \eta = \overline{\rho}\left(\frac{F_k}{U_k}-\gamma_x\right)\left(\frac{L_y}{2k\pi}\right)^2.
\end{equation}

A confidence interval on the value of~$U_k$ can straightforwardly 
be obtained from the estimator~\eqref{eq:fourier1} using block averaging procedures.
The statistical uncertainty on the viscosity is then obtained with~\eqref{eq:viscosity}.
Since the coefficients~$U_k$ decrease very rapidly
as $|k|$~increases, the relative statistical errors increase rapidly as well. 
We therefore restricted ourselves to $|k|=1$ in our numerical simulations.

%--------- resultats -------------
\subsection{Numerical results}
\label{sec:num_res_2D_fluid}

In all cases, time averages were computed over $N_{\rm iter} \simeq 10^7$ iterations.

%-------- reponse lineaire ----------
\subsubsection{Linear response}

We first verify numerically the linearity of the longitudinal velocity
as a function of the magnitude~$\xi$ of the nongradient force. 
More precisely, we check that $|U_1|$ is constant, 
for the three forces~$F$ at hand, in the case when $(\gamma_x,\gamma_y)
= (1,1)$ (see Figure~\ref{fig:linearity}).
In the sequel, unless otherwise stated, the numerical results are 
obtained with $\xi = 0.1$.

The numerical results show that linear response is valid even for values of~$\xi$ large enough.
In fact, a more refined analysis (see~\cite{Remi}) 
shows that, even if no nonlinear effect can be observed on the longitudinal velocity
for the values of~$\xi$ we considered, nonlinear effects on the kinetic temperature
cannot be ignored for values of $|\xi| \geq 0.1$. 

Other numerical simulations (not reported here, see~\cite{Remi})
show that linear response actually applies in the case when $\gamma_x = 0$
although this case is not covered by our theoretical analysis.

\begin{figure}[h]
\psfrag{x}{$\xi$}
\psfrag{y}{$\left|U_1\right|$}
\psfrag{a}{$\gamma_x=0$}
\psfrag{b}{$\gamma_x=1$}
\centerline{\includegraphics[width=0.6\linewidth]{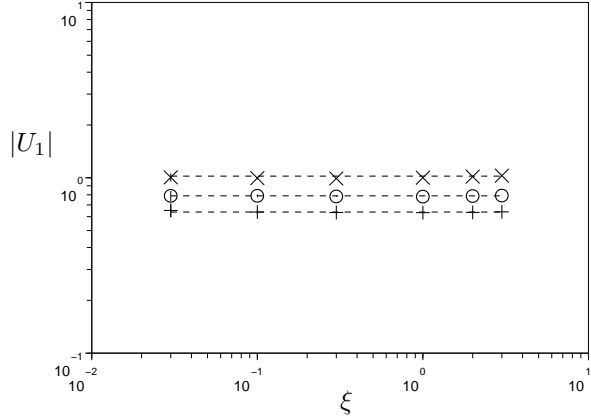}}
\caption{Value of $\left|U_1\right|$ as a function of $\xi$, for
  $(\gamma_x,\gamma_y) = (1,1)$ and the three nongradient forces
  at hand (piecewise constant perturbation $\times$; piecewise linear perturbation $+$; 
  sinusoidal perturbation $\circ$).}
\label{fig:linearity}
\end{figure}

%------------ validation closure ----------
\subsubsection{Validation of the closure}
\label{sec:validation_closure}

We present in Figures~\ref{fig:vel06},~\ref{fig:vel05} and~\ref{fig:vel07} 
the numerical approximations of the longitudinal velocity~$u_x$ and the off-diagonal
term of the stress tensor~$\sigma_{xy}$. The latter function is compared
to the quantity $-\eta u_x'$, where $\eta$ is obtained from~\eqref{eq:viscosity}, and
$u_x'$ is evaluated using a second order finite difference.
The good agreement between $\sigma_{xy}$ and~$-\eta u_x'$ validates the 
assumption~\eqref{eq:eta_constant}, the discrepancies resulting from statistical
fluctuations magnified by the numerical derivative, and also,
for the piecewise constant force, from the singularity
at $L_y/2$.

Besides, the velocity profile is consistent with~\eqref{eq:stokesfriction} (as can be checked by 
comparing the numerical solution and the solution of~\eqref{eq:stokesfriction} 
computed with the value
of~$\eta$ estimated from the simulation).

\begin{figure}[h]
  \psfrag{F}{$F$}
  \psfrag{U}{$u$}
  \psfrag{Y}{$Y$}
  \psfrag{v}{value}
\psfrag{S}{$\sigma_{xy}$}
\psfrag{D}{$-\eta u'$}
\begin{center}
\centerline{\includegraphics[width=0.5\linewidth]{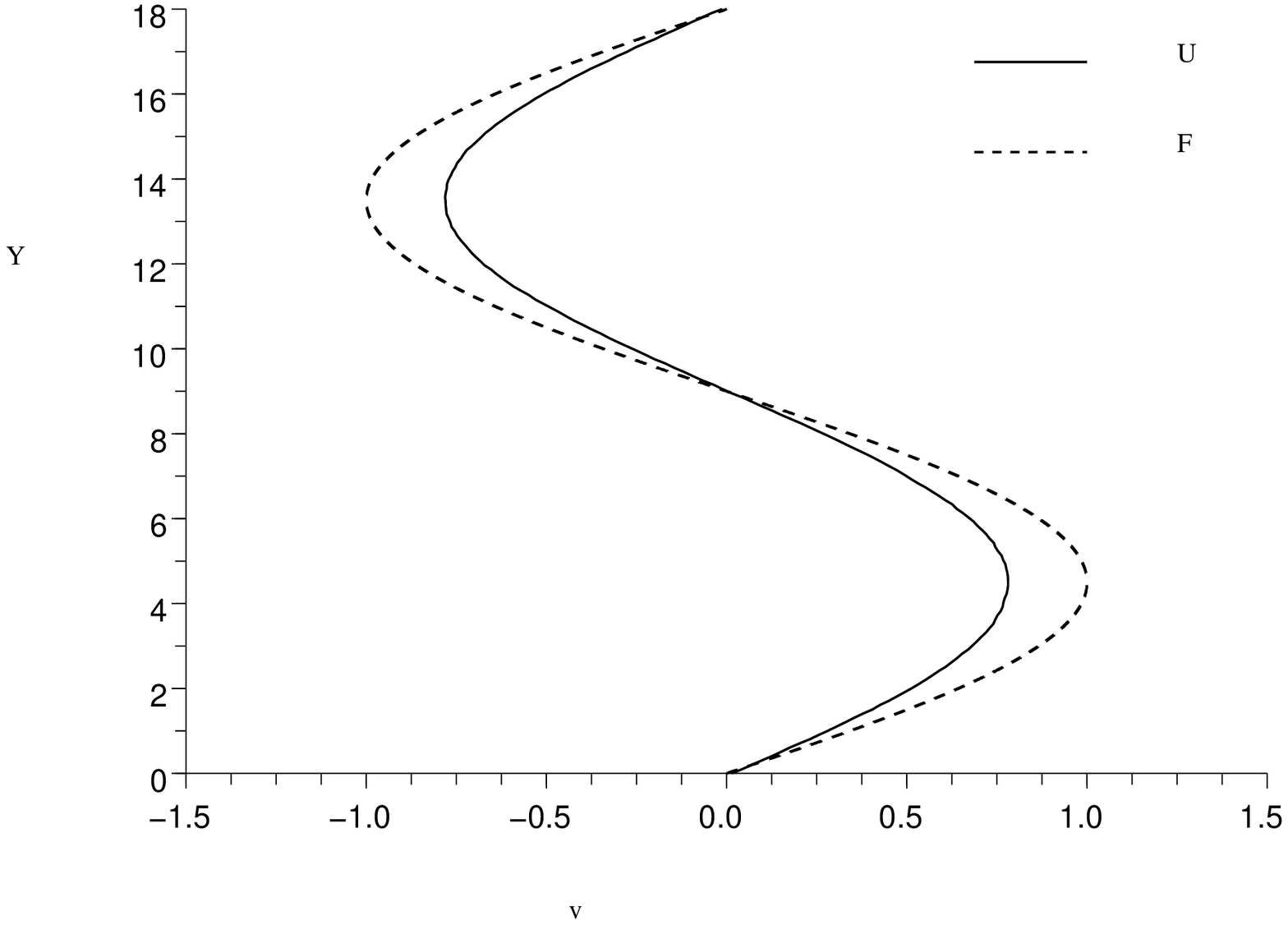}
\includegraphics[width=0.5\linewidth]{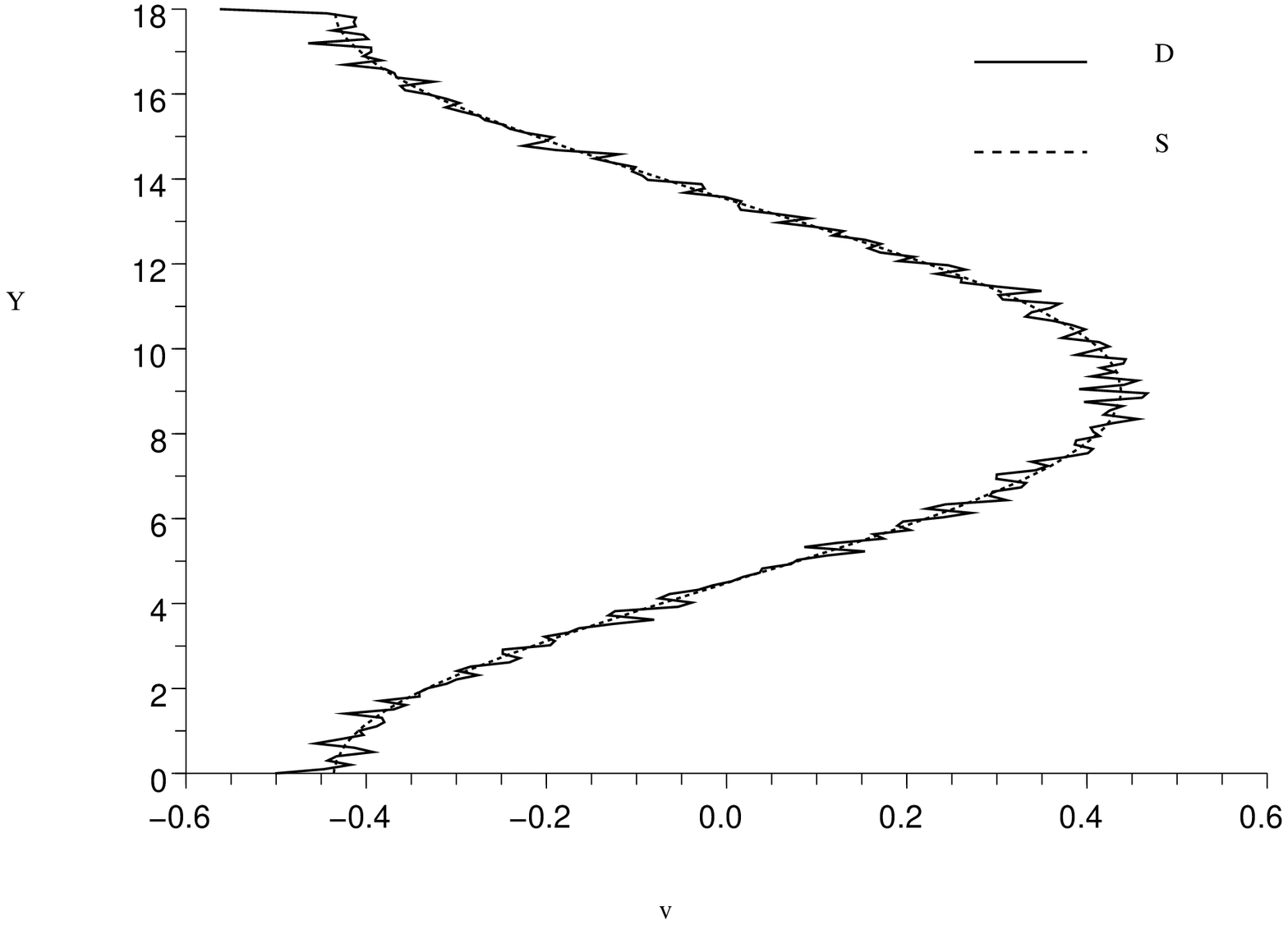}}
\end{center}\caption{Velocity profile and off diagonal component of the stress tensor 
for the sinusoidal nongradient force.}
\label{fig:vel06}
\end{figure} 

\begin{figure}[h]
\psfrag{F}{$F$}
\psfrag{U}{$u$}
\psfrag{Y}{$Y$}
\psfrag{v}{value}
\psfrag{S}{$\sigma_{xy}$}
\psfrag{D}{$-\eta u'$}
\begin{center}
\centerline{\includegraphics[width=0.5\linewidth]{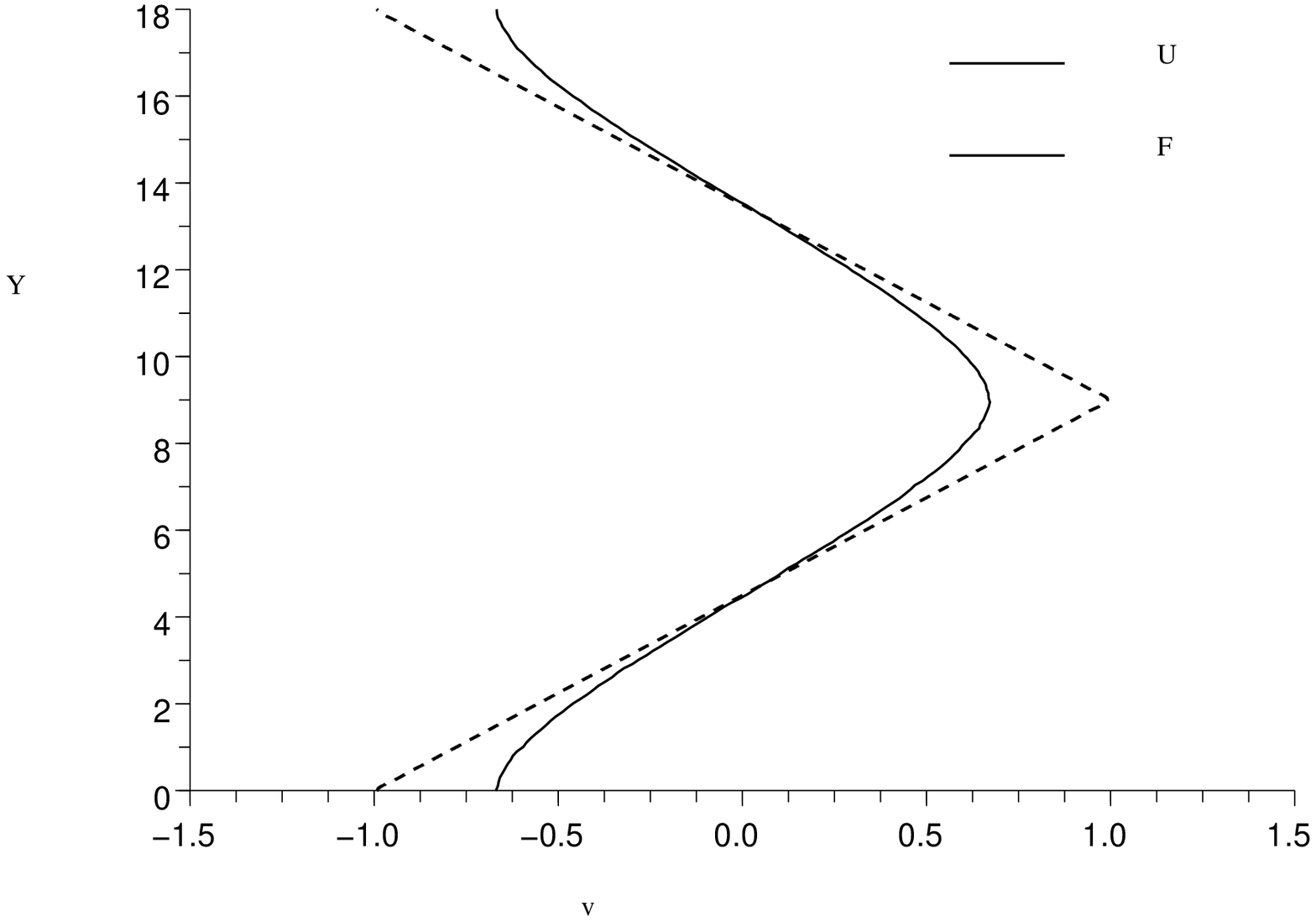}
\includegraphics[width=0.5\linewidth]{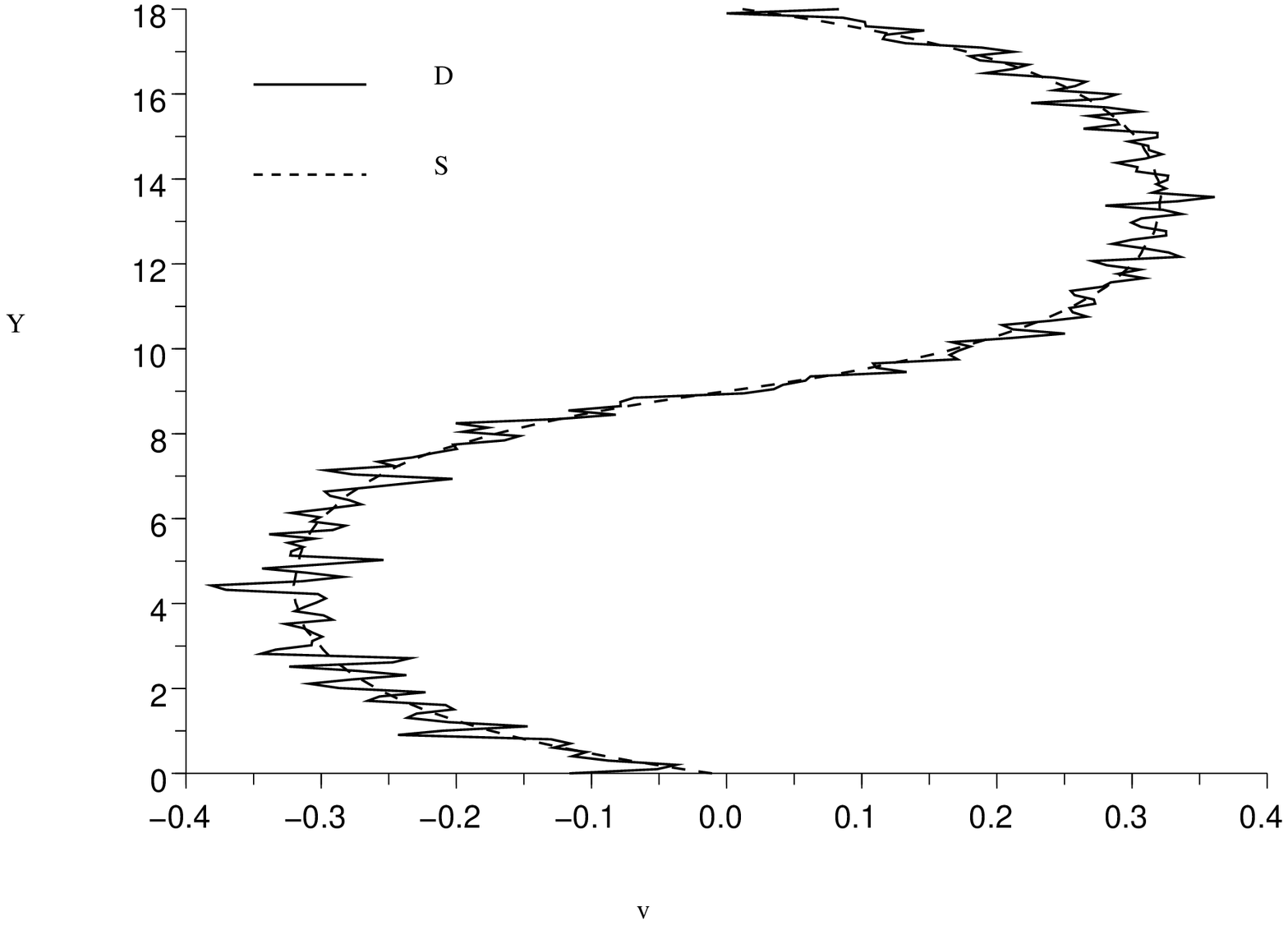}}
\end{center}\caption{Velocity profile and off diagonal component of the stress tensor 
for the piecewise linear nongradient force.}
\label{fig:vel05}
\end{figure} 

\begin{figure}[h]
  \psfrag{F}{$F$}
  \psfrag{U}{$u$}
  \psfrag{Y}{$Y$}
  \psfrag{v}{value}
  \psfrag{S}{$\sigma_{xy}$}
  \psfrag{D}{$-\eta u'$}
\begin{center}
\centerline{\includegraphics[width=0.5\linewidth]{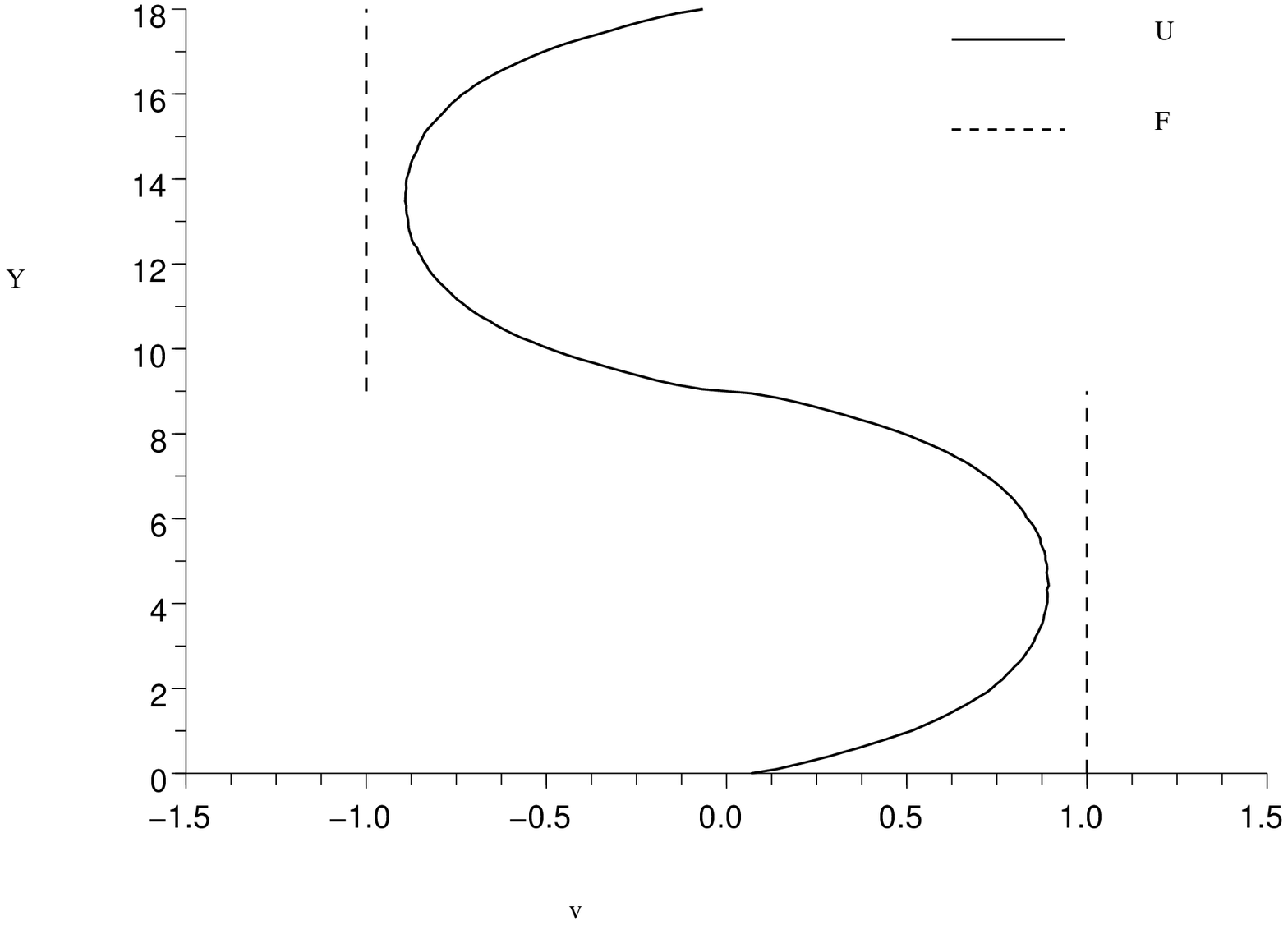}
\includegraphics[width=0.5\linewidth]{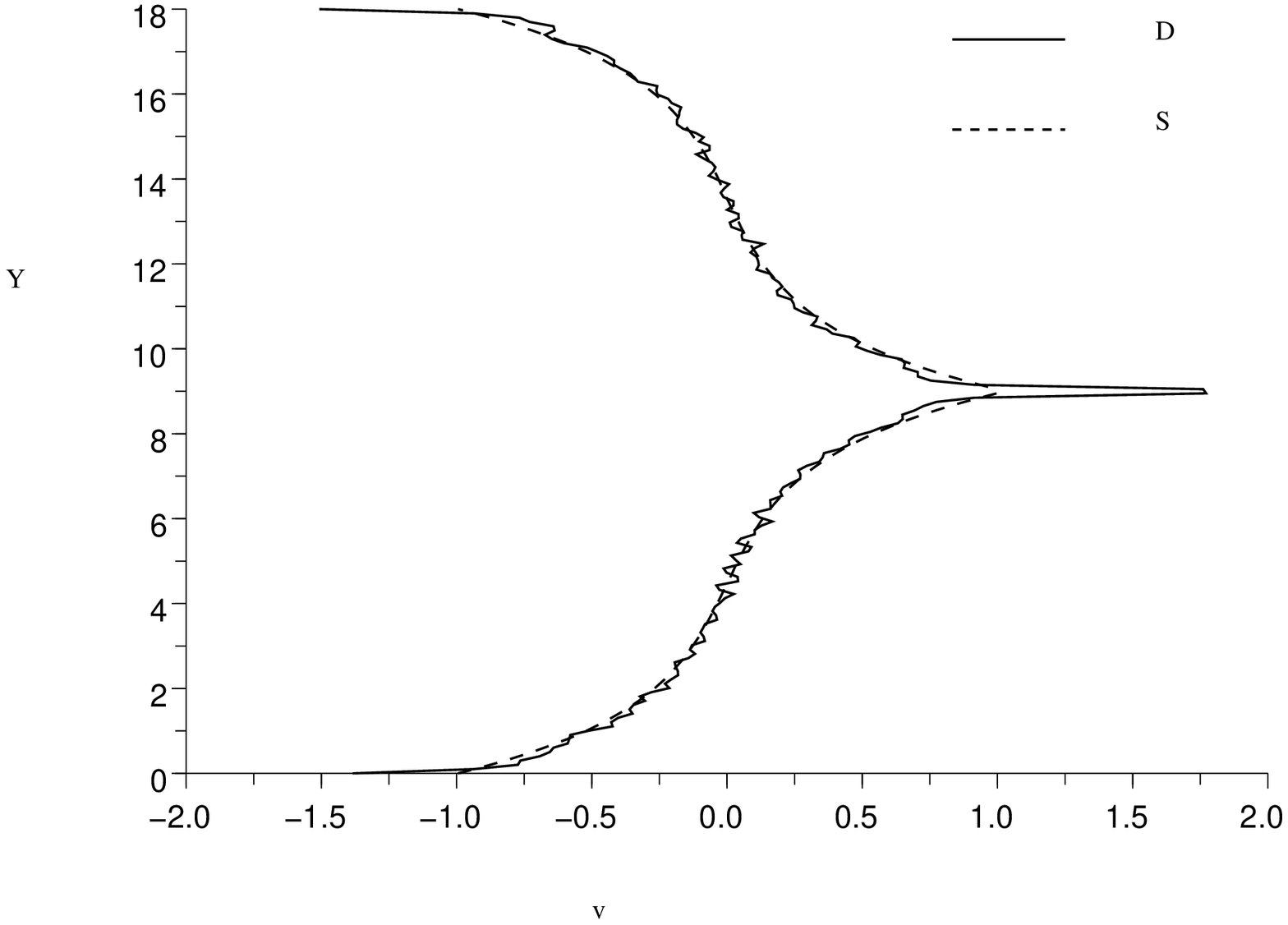}}
\end{center}
\caption{Velocity profile and off diagonal component of the stress tensor 
for the piecewise constant force.}
\label{fig:vel07}
\end{figure} 

%--------- dependance \gamma_y -------------
\subsubsection{Dependence and asymptotics in $\gamma_y$}
\label{sec:limite_num_gamma_y}

We first verify numerically that the velocity profiles 
converge to some limiting profile, and in fact that $U_1$
converges to some limiting value $U_1^\infty$ (see Figure~\ref{fig:limgy}).
We estimated $U_1^\infty$ by long simulations with 
$\alpha_y = 0$ in~\eqref{eq:numerical_scheme} 
(which amounts amounts to formally setting $\gamma_y$ to $+\infty$), and
computed the distance $|U_1^{\gamma_y}-U_1^\infty|$ as a function of
$\gamma_y$. A least square fit on the last computed values 
(in log-log scale) gives $|U_1^{\gamma_y}-U_1^\infty| \sim
\gamma_y^{-2.6}$. 

\begin{figure}[h]
  \psfrag{x}{$\gamma_y$}
  \psfrag{y}{$\left|U_1^{\gamma_y}-U_1^{\infty}\right|$}
\centerline{\includegraphics[width=0.7\linewidth]{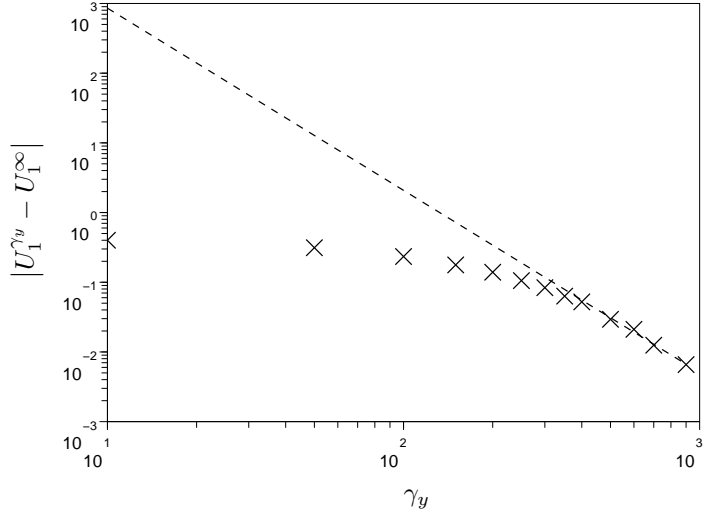}}
\caption{Convergence of the velocity profile for increasing values of the transverse
  friction~$\gamma_y$. The dashed line represents an affine fit in $\log$-$\log$ scale.}
\label{fig:limgy}
\end{figure}

We present in Figure~\ref{fig:viscGY} the dependence 
of the viscosity on the friction parameter~$\gamma_y$,
for a fixed value $\gamma_x = 1$. A mild dependence on $\gamma_y$ is observed in the limit
$\gamma_y \to 0$. The value obtained for $\gamma_y \to +\infty$ is on the other hand very different
from the limit obtained as $\gamma_y \to 0$. Note also that the viscosity seems to be an increasing
function of the transverse friction, which makes sense from a physical viewpoint.

\begin{figure}[h]
  \psfrag{x}{$\gamma_y$}
  \psfrag{y}{$\eta$}
\centerline{
\includegraphics[width=0.7\linewidth]{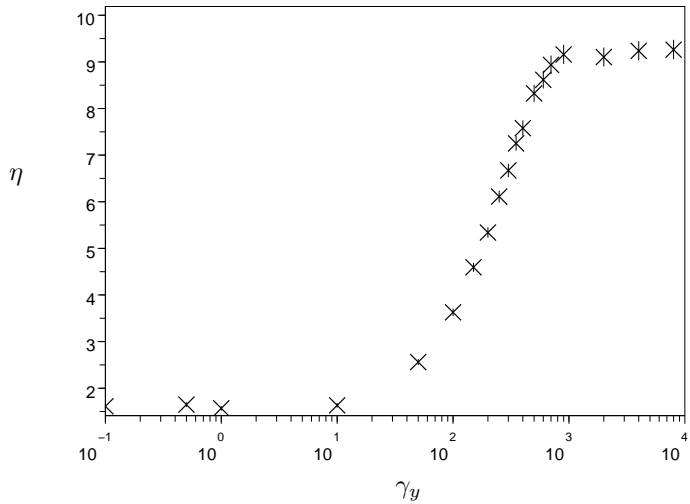}}
\caption{Shear viscosity $\eta$ as function of $\gamma_y$ in the case $\gamma_x=1$, for the sinusoidal nongradient force. 
}
\label{fig:viscGY}
\end{figure}

%--------- variance \gamma_x ----------
\subsubsection{Dependence and asymptotics in $\gamma_x$}
\label{sec:limite_num_gamma_x}

The discussion after Theorem~\ref{prop:gamma_x} suggests that the longitudinal 
velocity decreases as $1/\gamma_x$ as $\gamma_x \to +\infty$.
To observe numerically this behavior, it is necessary to increase the
magnitude of the nongradient force. Otherwise, the response is very
small (indeed, proportional to $\gamma_x^{-1}$) and relative statistical errors
are too large to obtain meaningful results.
We therefore computed the linear response
of the velocity for values of $\xi$ proportional to $\gamma_x$. This is
done by modifying the evolution on $p_x$ in~\eqref{eq:langevinheq} as follows:
\[
d p_{xi,t} = -\nabla_{q_{xi}} V(q_t) \, dt 
- \gamma_x \left( \frac{p_{xi,t}}{m} - \overline{\xi} F(q_{yi,t}) \right) 
dt + \sqrt{\frac{2\gamma_x}{\beta}} \, dW^{xi}_t.
\]
This amounts to replacing $\xi$ in~\eqref{eq:langevinheq} by
$\xi = \gamma_x \overline{\xi}$.
The resulting average velocity profile is $\gamma_x u_x$.

The results depicted in Figure~\ref{fig:limgx} show that the average
velocity, properly rescaled by $\gamma_x$, converges to the nongradient force,
as predicted by~\eqref{eq:limit_velocity_gamma_x}. 
The estimated convergence rate is 
$|\gamma_x U_1^{\gamma_x}-F_1| \sim \gamma_x^{-0.9}$.

The behavior of the corresponding viscosities cannot be predicted from the results
of Theorem~\ref{prop:gamma_x} (see the discussion at the end of 
Section~\ref{sec:infinite_gamma_x}). 
We therefore investigated numerically
this dependence, see Figure~\ref{fig:gammax}. The viscosity is more or
less constant for low values of $\gamma_x$, and increases for larger ones.

\begin{figure}[h]
  \psfrag{x}{$\gamma_x$}
  \psfrag{z}{\hspace{-0.2cm}$\gamma_x U_1$}
  \psfrag{y}{$\left|\gamma_x U_1^{\gamma_x}-F_1\right|$}
\centerline{
\includegraphics[width=0.5\linewidth]{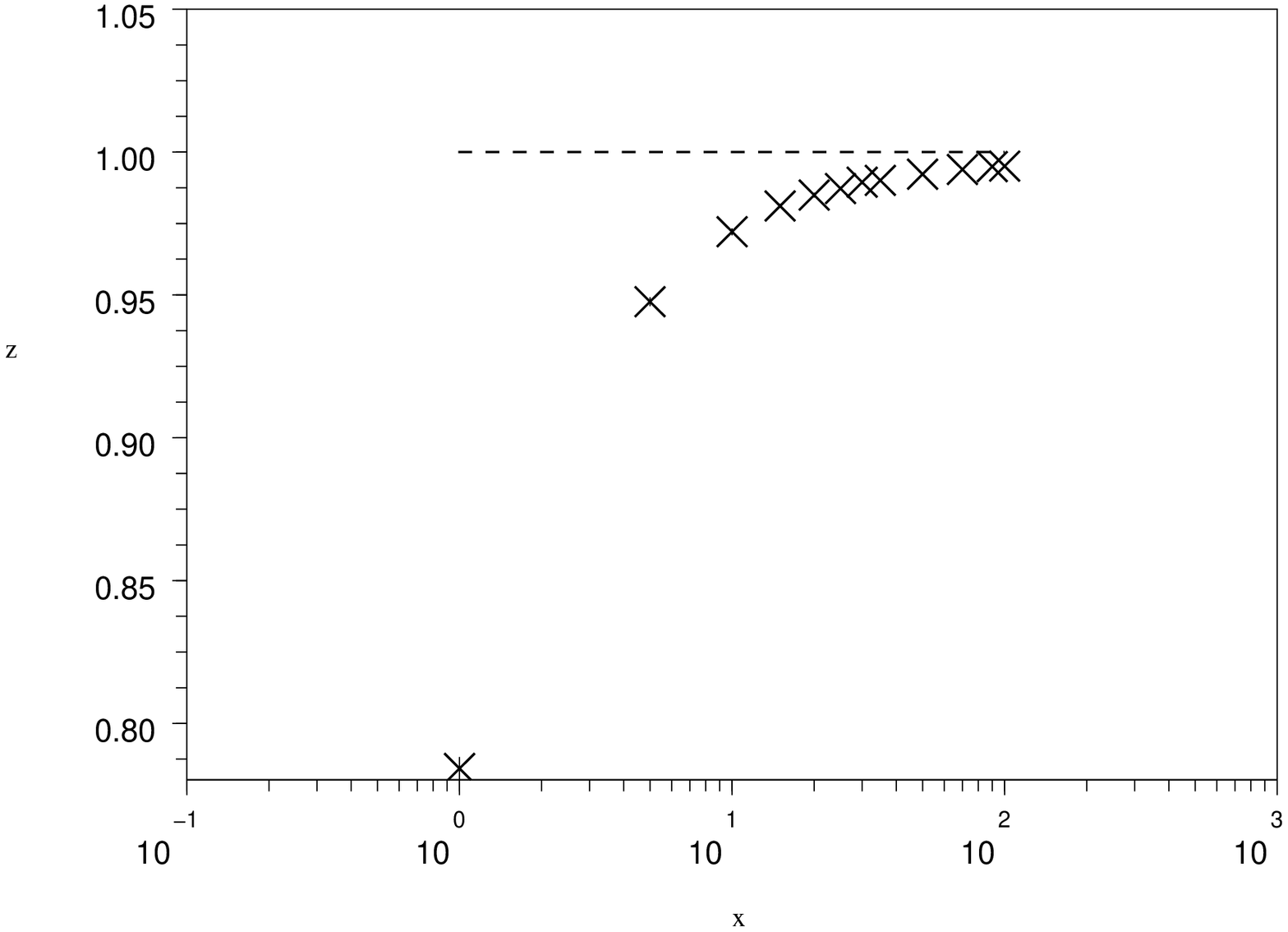}
\hspace{0.1cm}
\includegraphics[width=0.5\linewidth]{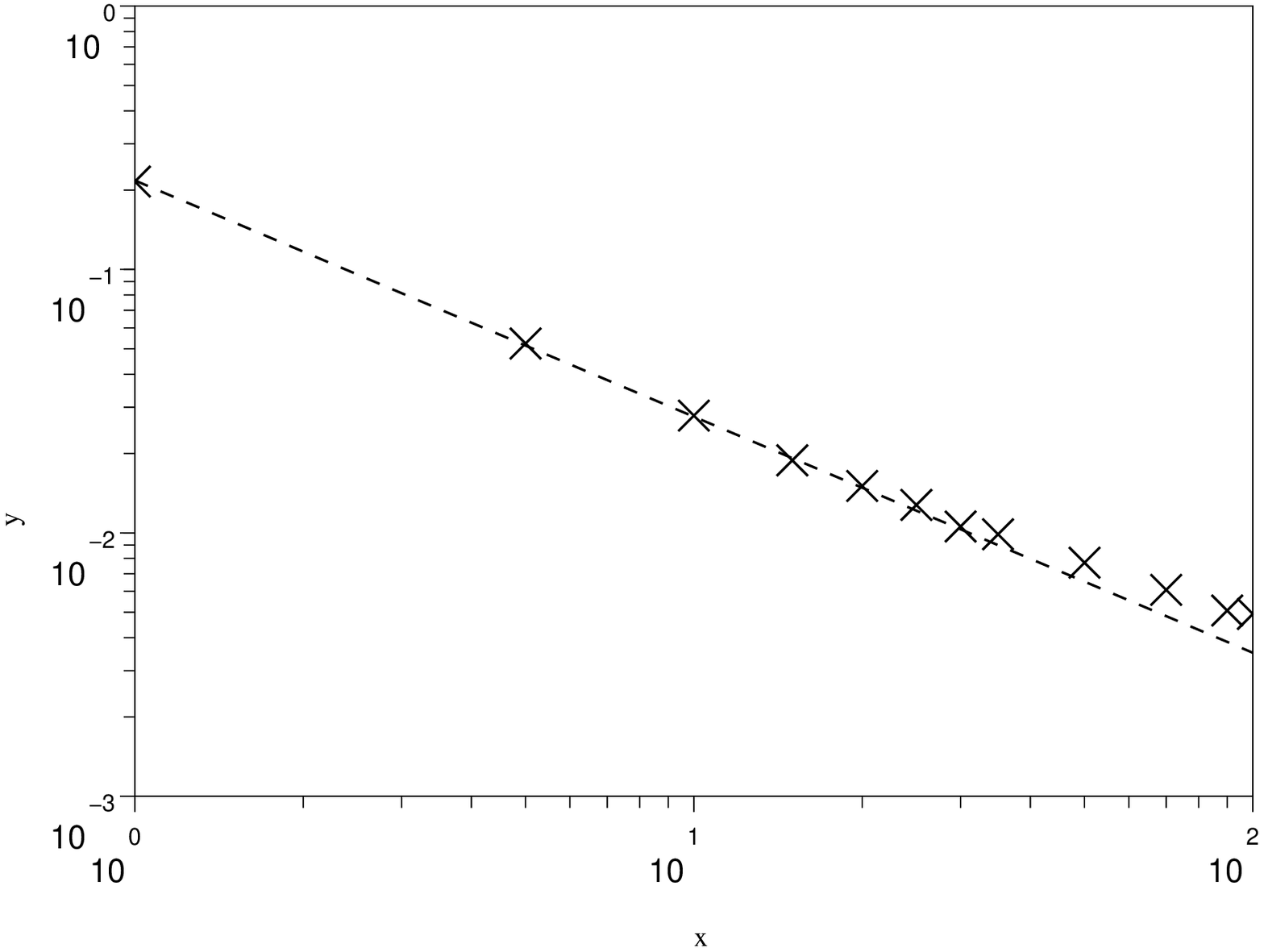}
}
\caption{Convergence of the velocity profile for increasing values of the friction~$\gamma_x$.
  The dashed line on the right picture represents an affine fit in $\log$-$\log$ scale.}
\label{fig:limgx}
\end{figure}

\begin{figure}[h]
  \psfrag{x}{$\gamma_x$}
  \psfrag{y}{$\eta$}
\centerline{
\includegraphics[width=0.5\linewidth]{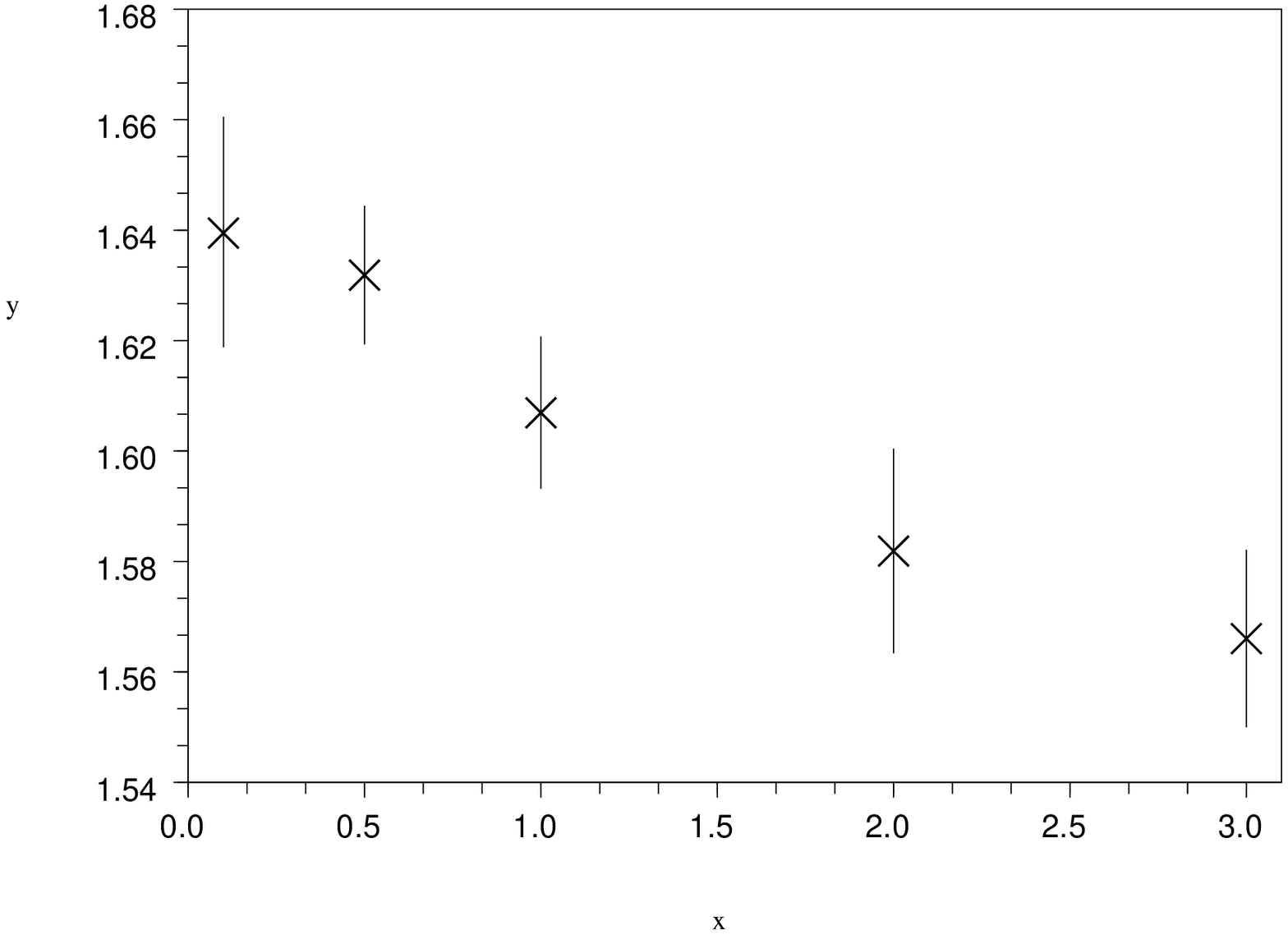}
\includegraphics[width=0.5\linewidth]{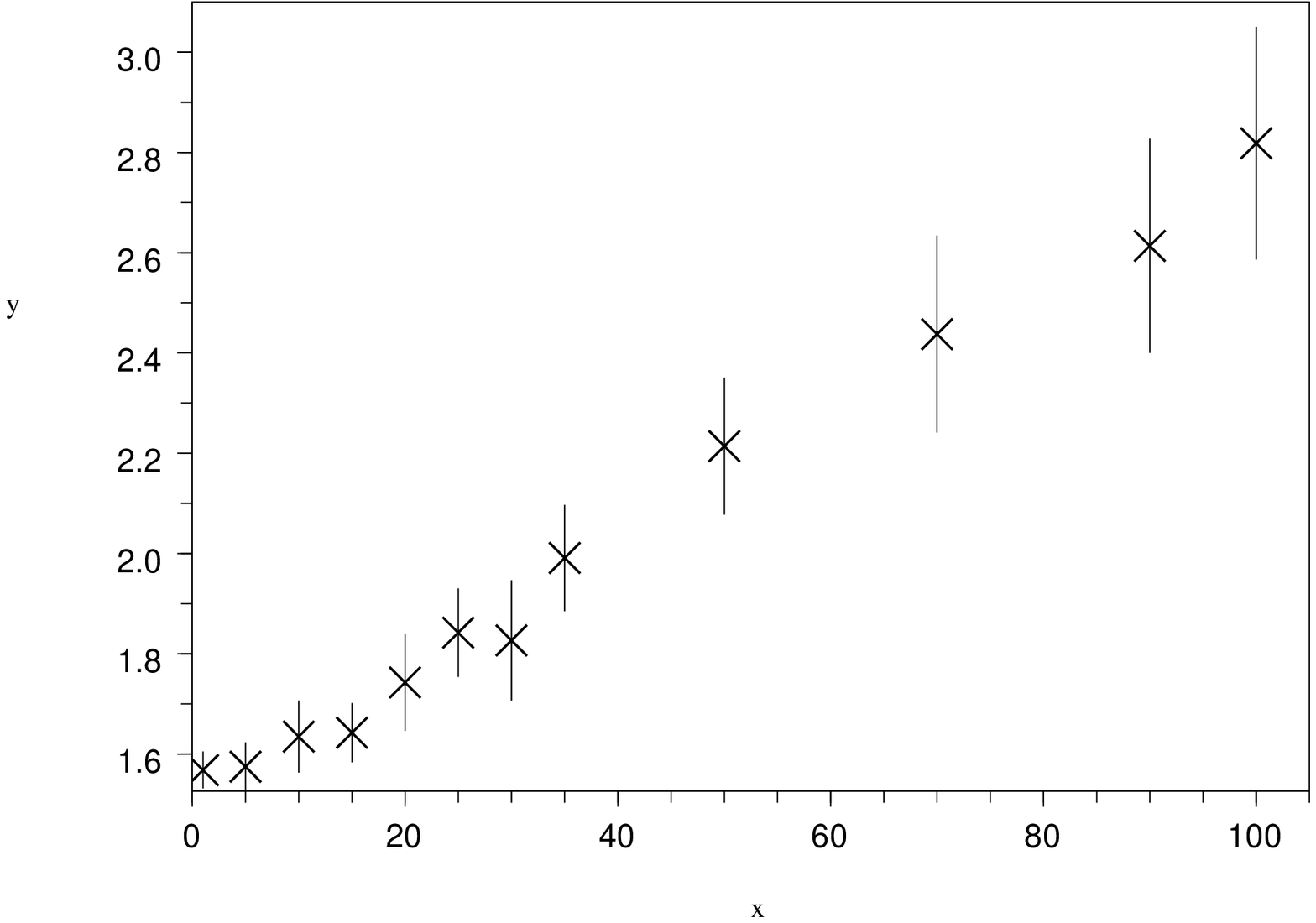}}
\caption{Shear viscosity $\eta$ as function of $\gamma_x$ in the case
  $\gamma_y=1$, for the sinusoidal nongradient force. Left: behavior for small 
  values of $\gamma_x$. Right: large $\gamma_x$ asymptotics.}
\label{fig:gammax}
\end{figure}

%----------------------
%     PROOFS
%----------------------

\section{Proof of the results}
\label{sec:proofs}

Unless otherwise stated, the norm $\| \cdot \|$ refers to the norm induced by the 
canonical scalar product on~$L^2(\psi_0)$.
The operator $\A_{\alpha,\rm thm}$ ($\alpha = x,y$), defined in~\eqref{eq:def_A_alpha},
can be rewritten as
\[
\A_{\alpha,\rm thm} = - \frac1\beta \sum_{i=1}^N \left(\partial_{p_{\alpha i}}\right)^* 
\partial_{p_{\alpha i}}.
\]
Note also that 
\[
\left[\partial_{p_{\alpha i}},\A_{\rm ham}\right] = \frac1m \partial_{q_{\alpha i}},
\]
where $[A,B] = AB-BA$ is the commutator of two operators.

%------------ Thm 1 ----------------
\subsection{Proof of Theorem~\ref{theo1}}
\label{sec:proof_psi_xi}

The existence and the uniqueness of the invariant measure which has a smooth density with 
respect to the Lebesgue measure for any $\xi \in \mathbb{R}$
is a standard result since the position space is compact and the forces 
are smooth. It suffices to use hypoellipticity arguments and take 
the kinetic energy as a Lyapunov function (see for 
instance~\cite{rey-bellet} for the general strategy, and \cite[Appendix~A]{pavliotis-stuart}
for the specific case under consideration). 
As a consequence, and recalling the definitions of the operators given in 
Section~\ref{sec:linear_response},
\[
\mathrm{Ker}(\L_0) = \mathrm{Span}(1) = \left\{ c \psi_0, c \in \mathbb{R} \right\}, 
\]
the vector space of constant functions on~$L^2(\psi_0)$. 
Note also that $\ker\left(\A_0\right) = \mathrm{Span}(1)$ 
by~\cite[Proposition~15]{Villani}.

The key result to prove the expansion~\eqref{eq:expansion_psi_xi}
is the following lemma (proved below).

\begin{lemm}
  \label{lem:bounded_op}
  The operators $\L_0^{-1}$ and $\L_0^{-1}\B^*$ are bounded operator on
  $\mathcal{H}$ (endowed with the $L^2(\psi_0)$ scalar product).
\end{lemm}

In view of this result, we can introduce 
\[
\xi^* = \left\| \L_0^{-1}\B^*|_\mathcal{H}\right\|^{-1}, 
\]
and define, for $k\geq 1$,
\[
\mathrm{f}_{k+1} = - \L_0^{-1}\B^* \mathrm{f}_k,
\]
with
\begin{equation}
  \label{eq:L0f1}
  \mathrm{f}_1 = -\L_0^{-1}\left(\B^* 1\right) = -\frac{\beta}{m} 
  \L_0^{-1}\left(\sum_{i=1}^N p_{xi} F(q_{yi})\right), 
\end{equation}
which is well-defined since the function $(q,p) \mapsto p_{xi} F(q_{yi})$
belongs to~$\mathcal{H}$ for all $i=1,\dots,N$.
The function $f_\xi$ in~\eqref{eq:expansion_psi_xi} 
is well defined for $|\xi| < \xi^*$ and 
a straightforward computation shows that 
\[
\L_{\xi} f_\xi = 0.
\]
The uniqueness of the invariant measure allows to conclude.

We now write the
 
\emph{Proof of Lemma~\ref{lem:bounded_op}.}
  We denote by $\| \cdot \|$ the $L^2(\psi_0)$-norm.
  Standard results of hypocoercivity show that $\A_0^{-1}$ is bounded on~$\mathcal{H}$ 
  (and in fact compact, see~\cite{Talay02,EH02,HN04,OttobrePhD}).
  Besides, for a smooth test function $\varphi$,
  \[
  \| \B \varphi \| \leq \| F \|_{L^\infty} \sum_{i=1}^N \| \partial_{p_{xi}} \varphi\|,
  \]
  while 
  \[
  \left\langle \varphi, \A_0 \varphi \right\rangle = - \frac1\beta \left( \gamma_x 
  \| \nabla_{p_x} \varphi \|^2 + \gamma_y \| \nabla_{p_y} \varphi \|^2\right).
  \]
  This shows that there exists a constant $C > 0$ such that, for any smooth test function 
  $\varphi$, 
  \[
  \| \B \varphi \|^2 \leq C \left|\left\langle \varphi, \A_0 \varphi \right\rangle\right|
  \leq C \| \varphi \| \, \|\A_0\varphi\|.
  \]
  In conclusion, for any $\varphi \in \mathcal{H}$, 
  \[
  \| \B \A_0^{-1}\varphi \|^2 \leq C \|\A_0^{-1}\varphi\| \, \| \varphi \|.
  \]
  Since $\mathrm{Ran}(\B) \subset \mathcal{H}$, this
  shows that $\B \A_0^{-1}$ is a bounded operator on $\mathcal{H}$.
  The same holds true for $\mathcal{L}_0^{-1} \B^* |_{\mathcal H}$, which is its adjoint 
  on $\mathcal{H} \subset L^2(\psi_0)$. \qquad\endproof

\begin{rem}
  \label{rmk:gamma=0}
  In the above proof, the fact that both $\gamma_x$ and $\gamma_y$ 
  are non-zero is a crucial assumption. 
  If for instance $\gamma_x = 0$, many arguments break down,
  and the proofs become much more technical and/or some results cannot be proved anymore. 
  This is due to the fact that the 
  the Lie algebra generated by $\{\A_{\rm ham},\partial_{p_{y1}},\dots,\partial_{p_{yN}}\}$ 
  may be different from
  the Lie algebra generated by $\{\A_{\rm ham},\partial_{p_{x1}},\dots,\partial_{p_{xN}},
  \partial_{p_{y1}},\dots,\partial_{p_{yN}}\}$. Indeed,
  \[
  [\A_{\rm ham},\partial_{p_{yi}}] = \partial_{q_{yi}},
  \qquad
  [\A_{\rm ham},\partial_{q_{yi}}] = \sum_{j=1}^N \partial_{q_{yi},q_{yj}}^2 V 
  \cdot \partial_{p_{yj}} + \sum_{j=1}^N \partial_{q_{yi},q_{xj}}^2 V \cdot \partial_{p_{xj}}.
  \]
  Possibly, iterated commutators should be computed as well. 
  Additional assumptions on the potential are required to infer that $\partial_{p_{xi}}$
  is in the Lie algebra. 
  This amounts to assuming that the coupling between the $x$ and the $y$
  directions is strong enough. 
  To our knowledge, the only cases where such arguments could be used
  are one-dimensional atom chains, for which the simple
  geometric structure of the system is of paramount importance
  to show that the Lie algebra has full rank. Obtaining hypocoercivity estimates
  is more challenging and imposes further restrictions on the interactions
  (see~\cite{EPR99}).
\end{rem}

%--------- conservation moment angulaire -------------
\subsection{Proof of Proposition~\ref{prop:local_conservation}}

Corollary~\ref{replin} shows that
\[
\begin{aligned}
\lim_{\xi \to 0} \frac{\left\langle \A_0 U_x^\varepsilon(Y,\cdot)\right\rangle_\xi}{\xi}
& = -\frac{\beta}{m} \left\langle U_x^\varepsilon(Y,q,p), \sum_{i=1}^N p_{xi} F(q_{yi}) 
\right\rangle_{L^2(\psi_0)} \\
& = -\frac{L_y}{mN} \sum_{i=1}^N 
\int_{\D^N} \chi_\varepsilon(q_{yi}-Y) F(q_{yi}) \, \overline{\psi}_0(q) \, dq,
\end{aligned}
\]
where $\overline{\psi}_0(q) \, dq = Z_q^{-1} \mathrm{e}^{-\beta V(q)} \, dq$ 
is the marginal of the canonical measure in the 
$q$~variable.
The integrand in the last equation depends only on one variable $q_{yi}$. By translation
invariance of the system, the marginal distribution in the $q_{yi}$ 
variable of $\overline{\psi}_0(q) \, dq$ is the uniform distribution on $L_y \mathbb{T}$.
Therefore, for any $i=1,\dots,N$,
\[
\int_{\D^N} \chi_\varepsilon(q_{yi}-Y) F(q_{yi}) \, \overline{\psi}_0(q) \, dq
= \frac{1}{L_y} \int_0^{L_y} \chi_\varepsilon(y-Y) F(y) \, dy \longrightarrow F(Y)
\]
as $\varepsilon \to 0$, so that
\begin{equation}
  \label{eq:direct_limit}
  \lim_{\varepsilon \to 0}
  \lim_{\xi \to 0} \frac{\left\langle \A_0 U_x^\varepsilon(Y,\cdot)\right\rangle_\xi}{\xi}
  = -\frac{1}{m} F(Y).
\end{equation}

Now, a simple computation shows that
\[ 
\begin{aligned}
\A_0 U_x^\varepsilon(Y,q,p) & =
\frac{1}{\rho L_x} \left( \sum_{i=1}^N \frac{p_{xi}p_{y_i}}{m}\partial_{q_{y_i}}
\chi'_{\varepsilon}\left(q_{y_i}-Y\right) 
 -\chi_{\varepsilon}\left(q_{y_i}-Y\right)\partial_{q_{xi}} V(q) \right) \\
& \ \ - \frac{\gamma_x}{m} U_x^\varepsilon(Y,q,p).
\end{aligned}
\]
The sum on the right-hand side can be decomposed into two contributions,
one proportional to the kinetic part of the off-diagonal part of the stress tensor,
and the other one arising solely from interaction forces.
The first contribution can be written as 
\[
-\frac{\partial\Sigma^\varepsilon_{xy,\rm kin}(Y,q,p)}{\partial Y} = 
-\frac{d}{dY}\left ( \frac{1}{\rho L_x} \sum_{i=1}^N \frac{p_{xi}p_{y_i}}{m} 
\chi_{\varepsilon}\left(q_{y_i}-Y\right) \right).
\]
For the second part, we first use the pairwise character of the interactions
to write
\[
\partial_{q_{xi}} V(q) = \sum_{i \neq j} 
\mathcal{V}'(|q_i-q_j|)\frac{ q_{xi}-q_{xj}}{|q_i-q_j|},
\]
and then symmetrize the resulting expression as 
\begin{align*}
  \sum_{i=1}^N \chi_{\varepsilon}\left(q_{y_i}-Y\right)\partial_{q_{xi}} V(q) &= 
  \sum_{i\neq j} \chi_{\varepsilon}\left(q_{y_i}-Y\right)\mathcal{V}'(|q_i-q_j|)
  \frac{ q_{xi}-q_{xj}}{|q_i-q_j|}, \\ 
  & = \! \! \! \sum_{1 \leq i < j \leq N} \big(\chi_{\varepsilon}\left(q_{yi}-Y\right)-
  \chi_{\varepsilon}\left(q_{yj}-Y\right)\big)
  \mathcal{V}'(|q_i-q_j|) \frac{ q_{xi}-q_{xj}}{|q_i-q_j|}.
\end{align*}
The second contribution finally reads
\[
-\frac{\partial\Sigma^\varepsilon_{xy,\rm pot}(Y,q,p)}{\partial Y} = 
\frac{d}{dY}\left( \frac{1}{\rho L_x} \sum_{1 \leq i < j \leq N} 
\mathcal{V}'(|q_i-q_j|)\left(\frac{ q_{xi}-q_{xj}}{|q_i-q_j|}\right)
\int_{q_{yj}}^{q_{yi}} \chi_{\varepsilon}(s-Y) \, ds \right).
\]
In conclusion, it holds
\begin{equation}
  \label{eq:A0_U}
  \A_0 U_x^\varepsilon(Y,q,p) = -\frac{1}{\rho} 
  \frac{\partial \Sigma^\varepsilon_{xy}(Y,q,p)}{\partial Y} 
  - \frac{\gamma_x}{m} U_x^\varepsilon(Y,q,p).
\end{equation}

Combining the latter result and~\eqref{eq:direct_limit} leads to
\begin{equation}
  \label{eq:final_eq_local_conservation}
  \lim_{\varepsilon \to 0}
  \lim_{\xi \to 0} \frac{1}{\xi} \left( 
  \frac1\rho \frac{\partial \left\langle\Sigma^\varepsilon_{xy}(Y,\cdot)\right\rangle_\xi}
             {\partial Y}
  + \frac{\gamma_x}{m} \left\langle U_x^\varepsilon(Y,\cdot) \right\rangle_\xi 
\right) = \frac{1}{m} F(Y).
\end{equation}
Now, Corollary~\ref{replin} shows that the limit 
\[
u_x^\varepsilon(Y) := \lim_{\xi \to 0} \frac{\left\langle U_x^\varepsilon(Y,\cdot) \right\rangle_\xi}
{\xi} = \frac{L_y}{Nm} \sum_{i=1}^N 
\int_{\D^N \times \R^{2N}} p_{xi} \chi_\varepsilon(q_{yi}-Y) \mathrm{f}_1(q,p) 
\psi_0(q,p) \, dq \, dp
\]
is well defined. 
By hypoellipticity, the function $\mathrm{f}_1$ belongs to $C^\infty(\D^N \times \mathbb{R}^{2N})$.
The limit $\varepsilon \to 0$ of the right-hand side is therefore well defined
and
\[
u_x(Y) = \frac{L_y}{Nm} \sum_{i=1}^N 
\int_{\D^{N-1}\times L_x\mathbb{T} \times \R^{2N}} 
p_{xi} (f^1\psi_0)(q_1,\dots,q_{i-1},q_{xi},Y,q_{i+1},\dots,q_N,p) \, dq_{1:i:N} \, dq_{xi} \, dp,
\]
where $dq_{1:i:N} = dq_1\dots dq_{i-1} \, dq_{i+1} \dots dq_N$.
A similar reasoning holds for $\sigma_{xy}$. Passing to the limit 
in~\eqref{eq:final_eq_local_conservation},
\[
\frac1\rho \frac{\partial \sigma_{xy}(Y)}{\partial Y}
+ \frac{\gamma_x}{m} u_x(Y)
= \frac{1}{m} F(Y),
\]
which is~\eqref{eq:local_conservation}.

%--------- gamma_y tend vers l'infini -------------------
\subsection{Proof of Theorem~\ref{prop:gamma_y}}
\label{sec:proof_gamma_y}

To simplify the notation, we set $m = 1$ in this section, but the proof can 
be straightforwardly modified to account for more general masses.
Note first that the solution of~\eqref{eq:eq_Poisson_y} is well defined
for any $\gamma_y > 0$ by the Fredholm alternative (since $\A_{0}(\gamma_y)$ has a compact
resolvent on $\mathcal{H} = L^2(\psi_0) \cap \{ 1 \}^\perp$,
and the right-hand side of the equation is orthogonal to $\mathrm{Vect}(1)$).

We start by formal computations providing possible expressions of $f^0,f^1$,
and then prove rigorously the convergence result stated in Theorem~\ref{prop:gamma_y}.
To this end, we need some intermediate uniform hypocoercivity result.

\paragraph{Formal asymptotic expansion in $\gamma_y$}
We consider the following \emph{ansatz} for the solution~$f_{\gamma_y}$:
\[
f_{\gamma_y} = f^0 + \frac{1}{\gamma_y} f^1 + \frac{1}{\gamma_y^2} f^2 + \cdots
\]
and rewrite the operator $\A_0(\gamma_y)$ as the sum 
$\A_0(\gamma_y) = T_0 + \gamma_y \A_{y, \rm thm}$.
The kernel of the operator $\A_{y, \rm thm}$ on $L^2(\psi_0)$ is
\[
\ker(\A_{y, \rm thm}) = \Big\{ g \in L^2(\psi_0) \ \Big| \ g = g(q,p_x) \Big\}.
\]
This is a consequence of the equality
\[
\left \langle g, \A_{y, \rm thm}g \right\rangle_{L^2(\psi_0)} 
= -\frac1\beta\| \nabla_{p_y} g \|^2
\]
and the fact that the Gaussian measure (in the $p_y$ variable) satisfies
a Poincar\'e inequality. 
Identifying terms with the same powers of $\gamma_y$ 
in~\eqref{eq:eq_Poisson_y}, the following hierarchy is obtained:
\begin{equation}
  \label{eq:hierarchy_y}
  \left \{ \begin{aligned}
    \A_{y, \rm thm}f^0 & = 0, \\
    T_0 f^0 + \A_{y, \rm thm} f^1 & = -\sum_{i=1}^N p_{xi} G(q_{yi}), \\
     T_0 f^1 + \A_{y, \rm thm} f^2 & = 0.
  \end{aligned} \right.
\end{equation}
The first equation shows that $f^0 \equiv f^0(q,p_x)$. 
The second equation can then be rewritten as
\[
\A_{y, \rm thm} f^1(q,p) =  -p_y \cdot \nabla_{q_y} f^0(q,p_x) - \sum_{i=1}^N p_{xi} G(q_{yi}) - 
\mathcal{T}_{q_y}f^0(q,p_x).  
\]
where 
\[
\mathcal{T}_{q_y} = p_x \cdot \nabla_{q_x} - \nabla_{q_x}V(q_x,q_y) \cdot \nabla_{p_x} 
+ \gamma_x \A_{x, \rm thm}
\]
is an operator parameterized by $q_y \in (L_y \mathbb{T})^N$, 
and acting on the Hilbert space $L^2(\Psi_{q_y})$, where
\[
\Psi_{q_y}(q_x,p_x) = Z_{q_y}^{-1} \exp\left(-\beta \left( V(q_x,q_y) + \frac{p_x^2}{2m}\right) 
\right).
\]
Setting
\[
f^1 = \widetilde{f}^1 + p_y \cdot \nabla_{q_y} f^0, 
\]
it holds
\[
\A_{y, \rm thm} \widetilde{f}^1(q,p) = -\sum_{i=1}^N p_{xi} G(q_{yi}) - 
\mathcal{T}_{q_y}f^0(q,p_x).
\]
Since the right-hand side does not depend on $p_y$, 
the solvability condition for this equation is that the right-hand side vanishes.
Besides, by the same results which allow to prove that $\A_0$ has compact resolvent, the operator~$\mathcal{T}_{q_y}$, considered as an operator on $L^2(\Psi_{q_y})$, has a compact resolvent on  $\ker(\mathcal{T}_{q_y})^\perp = \{ 1 \}^\perp$ (where the orthogonality is with respect to the canonical scalar product on~$L^2(\Psi_{q_y})$; see~\cite[Proposition~15]{Villani} for a proof of the latter equality). In addition, by hypocoercivity, $\mathcal{T}_{q_y}^{-1}$ is bounded on $H^1(\Psi_{q_y}) \cap \ker(\mathcal{T}_{q_y})^\perp$. Therefore, $\mathcal{T}_{q_y}^{-1}(p_{xi})$ is well defined by the Fredholm alternative. By linearity,
\[
f^0(q,p) = -\sum_{i=1}^N G (q_{yi})\mathcal{T}_{q_y}^{-1}(p_{xi}),
\]
and
\[
f^1(q,p) = p_y \cdot \nabla_{q_y} f^0(q,p) + \widetilde{f}^1(q,p_x), 
\]
provide admissible solutions for the first two levels of the 
hierarchy~\eqref{eq:hierarchy_y}. The function~$\widetilde{f}^1$ will be made precise below (see~\eqref{eq:def_f1_tilde}).

The function $f^0$ is in $H^1(\psi_0)$. 
To show that the function $f^1-\widetilde{f}^1$ is indeed well defined, 
it is enough to show that $\partial_{q_{yi}} \left[ \mathcal{T}_{q_y}^{-1}(p_{xk})\right]$
is well defined. This, in turn, follows from the following equality
for any function $\varphi = \varphi(q_x,p_x)$:
\begin{equation}
  \label{eq:deriver_T_moins_un}
  \partial_{q_{yi}} \left(\mathcal{T}_{q_y}^{-1}\right) \varphi
  = \left(\mathcal{T}_{q_y}^{-1}\right) \left[
    \sum_{j=1}^N \partial^2_{q_{yi},q_{xj}} V(q_x,q_y) \partial_{p_xj}
    \right]\left(\mathcal{T}_{q_y}^{-1}\right) \varphi.
\end{equation}
The operators $\partial_{p_{x j}}\left(\mathcal{T}_{q_y}^{-1}\right)$
are bounded on $L^2(\Psi_{q_y})\cap \{1\}^\perp$ for $j=1,\dots,N$ since
\[
\left\langle\varphi,\mathcal{T}_{q_y}\varphi\right\rangle_{L^2(\Psi_{q_y})} 
= -\frac{\gamma_x}{\beta}\| \nabla_{p_x}\varphi\|_{L^2(\Psi_{q_y})}^2,
\]
which implies, for $\| \varphi \|_{L^2(\Psi_{q_y})} \leq 1$,
\[
\left\| \nabla_{p_x} \left(\mathcal{T}_{q_y}^{-1}\right)\varphi\right\|_{L^2(\Psi_{q_y})}^2
\leq \frac{\beta}{\gamma_x} \|\varphi\|_{L^2(\Psi_{q_y})} \, 
\left\| \left(\mathcal{T}_{q_y}^{-1}\right)\varphi\right\|_{L^2(\Psi_{q_y})} 
\leq \frac{\beta}{\gamma_x} \left\| \mathcal{T}_{q_y}^{-1}\right\|.
\]
In conclusion, $f^1-\widetilde{f}^1 \in H^1(\psi_0)$. In addition, by hypoellipticity, the functions $f^0,f^1-\widetilde{f}^1$ are $C^\infty$ when $G$ is smooth.

\paragraph{Uniform hypocoercivity estimates}
Let us show that the operator $\A_0(\gamma_y)$ is uniformly hypocoercive
for $\gamma_y$ large enough (say, $\gamma_y \geq \gamma_x$), provided the domain of the operator is restricted to functions with vanishing average with respect to the Gibbs measure in the $p_y$ variable. To this end, we decompose $\A_0(\gamma_y)$ as
\[
\A_0(\gamma_y) = \A_0(\gamma_x)  + (\gamma_y-\gamma_x)\A_{y, \rm thm}.
\]
Following the proof of Theorem~6.2 in~\cite{HairerPavliotis08}, 
it can be shown that there exists $\kappa > 0$
such that, for all smooth functions $u \in \mathcal{H}$,
\[
-\left\langle\left\langle u, \A_0(\gamma_x)u \right\rangle\right\rangle
\geq \kappa \left\langle\left\langle u, u \right\rangle\right\rangle,
\]
where the norm induced by 
$\left\langle\left\langle \cdot, \cdot\right\rangle\right\rangle$ is equivalent
to the $H^1(\psi_0)$ norm
\[
\| u \|^2_{H^1(\psi_0)} = \| u \|^2 + \| \nabla_p u \|^2 + \| \nabla_q u\|^2.
\]
More precisely, $\left\langle\left\langle \cdot, \cdot\right\rangle\right\rangle$
is the  bilinear form defined by
\[
\left\langle\left\langle u, v\right\rangle\right\rangle = a\left\langle u, v\right\rangle
+ b \left\langle \nabla_p u,\nabla_pv\right\rangle + \langle \nabla_p u,
\nabla_q v\rangle +  \langle \nabla_q u,
\nabla_p v\rangle + b\langle \nabla_q u,
\nabla_q v\rangle,
\]
with appropriate coefficients $a \gg b \gg 1$.
It follows that there exists $C> 0$ independent of $\gamma_y$ such that
\[
C \left\| u \right\|_{H^1(\psi_0)}^2 - (\gamma_y-\gamma_x)
\left\langle\left\langle u, \A_{y, \rm thm} u \right\rangle\right\rangle
\leq -\left\langle\left\langle u, \A_0(\gamma_y)  u \right\rangle\right\rangle.
\]
Let us now show that 
\begin{equation}
\label{eq:estimate_positiviy_A}
-\left\langle\left\langle u, \A_{y, \rm thm} u \right\rangle\right\rangle 
\geq 0
\end{equation}
for functions $u$ in an appropriate subspace of $H^1(\psi_0)$.
Using the commutation relations $[\partial_{p_{\alpha,i}}, \partial_{p_{\alpha',j}}^*] = \beta \delta_{\alpha,\alpha'}\delta_{ij}$ ($\alpha,\alpha' \in \{ x,y \}$), a simple computation shows
\[
\begin{aligned}
\left\langle\left\langle u, \sum_{i=1}^N\left(\partial_{p_{yi}}\right)^*\partial_{p_{yi}} u \right\rangle\right\rangle 
& = \sum_{i=1}^N (a+\beta b) \| \partial_{p_{y_i}} u \|^2 + b \|\nabla_p \partial_{p_{yi}} u\|^2  \\
& \quad + b \|\nabla_q \partial_{p_{yi}} u\|^2 + 2\langle \nabla_q \partial_{p_{yi}} u, \nabla_p \partial_{p_{yi}} u\rangle + \beta \langle \partial_{q_{yi}} u, \partial_{p_{yi}} u\rangle \\
& \geq \sum_{i=1}^N \left(a+\beta \left(b-\frac12\right)\right) \| \partial_{p_{y_i}} u \|^2 + (b-1) \|\nabla_p \partial_{p_{yi}} u\|^2 \\
& \qquad + (b-1) \|\nabla_q \partial_{p_{yi}} u\|^2 - \frac{\beta}{2} \| \partial_{q_{yi}} u \|^2.
\end{aligned}
\]
Summing on $i \in \{ 1,\dots,N\}$, the quantity~\eqref{eq:estimate_positiviy_A} is seen to be non-negative for an appropriate choice of constants $a \gg b \gg 1$ provided there exists a constant $A > 0$ such that, for all $i = 1,\dots,N$,
\begin{equation}
  \label{eq:condition}
  \| \partial_{q_{yi}}u \| \leq A \| \nabla_p \partial_{q_{yi}} u \|.
\end{equation}
This indeed implies
\[
\sum_{i=1}^N \| \partial_{q_{yi}} u \|^2 \leq A \sum_{i,j=1}^N \| \partial_{p_{yj}} \partial_{q_{yi}} u \|^2 = A \sum_{j=1}^N \| \nabla_{q_y} \partial_{p_{yj}}\|^2
\leq A \sum_{j=1}^N \| \nabla_{q} \partial_{p_{yj}}\|^2.
\]
Since the Gaussian measure satisfies a Poincar\'e inequality, the inequalities~\eqref{eq:condition} hold provided
\[
\forall i = 1,\dots,N, \qquad \int_{\R^N} \partial_{q_{yi}} u(q,p) \exp\left(-\beta\frac{p_{y}^2}{2}\right) \, dp_{y} = 0.
\]
Defining the closed subspace of $L^2(\psi_0) \cap \{ 1 \}^\perp$ 
\begin{equation}
\label{eq:def_H0}
\mathcal{H}_0 = \left\{ v \in H^1(\psi_0) \ \left| \ \overline{v}(q,p_x) = \left(\frac{2\pi}{\beta}\right)^{-N/2} \int_{\R^N} v(q,p) \, \exp\left(-\beta\frac{p_{y}^2}{2}\right) \, dp_{y} = 0 \right.\right\} \subset \mathcal{H},
\end{equation}
we conclude that, for any $u \in \mathcal{H}_0 \cap H^2(\psi_0)$,
\begin{equation}
  \label{eq:estimee_coercivite}
  C \left\| u \right\|_{H^1(\psi_0)}^2 
  \leq -\left\langle\left\langle u, \A_0(\gamma_y)  
  u \right\rangle\right\rangle.
\end{equation}
In particular, there exists a constant $K>0$ such that, for any $\gamma_y \geq \gamma_x$ and for any $u \in \mathcal{H}_0 \cap H^2(\psi_0)$, 
\[
\left\| \A_0(\gamma_y)^{-1}  u \right\|_{H^1(\psi_0)} \leq K \| u \|_{H^1(\psi_0)}.
\]
In fact, this inequality can be extended to functions in $\mathcal{H}_0$.

\paragraph{Proof of the limit~\eqref{eq:cv_Poisson_gamma_y}}
To prove~\eqref{eq:cv_Poisson_gamma_y}, we proceed as follows. Note first that
\[
-\A_0(\gamma_y) \left(f_{\gamma_{y}}- f^0-\gamma_{y}^{-1}f^1\right) = \frac{1}{\gamma_y} T_0 f^1,
\]
so that
\begin{equation}
\label{eq:reg}
f_{\gamma_{y}}- f^0-\gamma_{y}^{-1}f^1 = -\frac{1}{\gamma_y} \A_0(\gamma_y)^{-1} T_0 f^1.
\end{equation}
Since $\A_0(\gamma_y)^{-1}$ is bounded on $\mathcal{H}_0$, uniformly in $\gamma_y$ (see~\eqref{eq:estimee_coercivite}), it is sufficient to show that $T_0f^1 \in \mathcal{H}_0$. In view of the definition of $f^0$, the proof is then concluded by setting $\phi_i(q,p) = -\mathcal{T}_{q_y}^{-1}(p_{xi})$.

Let us first show that $\overline{T_0 f^1}(q,p_x) = 0$ (where $\overline{v}$ is defined in~\eqref{eq:def_H0}). This can be ensured by an appropriate choice of $\widetilde{f}^1$. Note first that
\[
T_0 f^1 = p_y \cdot \nabla_{q_y} \widetilde{f}^1 + \mathcal{T}_{q_y} f^1 + \left(p_y\cdot \nabla_{q_y} - \nabla_{q_y} V \cdot \nabla_{p_y}\right)f^1+ \mathcal{T}_{q_y} \widetilde{f}^1. 
\]
The first two terms have a vanishing average with respect to $(2\pi)^{-N/2} \exp\left(-\beta\frac{p_{y}^2}{2}\right) \, dp_y$. Introducing
\[
g(q,p_x) = -(2\pi)^{-N/2} \int_{\R^N} \left(p_y\cdot \nabla_{q_y} - \nabla_{q_y} V \cdot \nabla_{p_y}\right)f^1 \, \exp\left(-\beta\frac{p_{y}^2}{2}\right) \, dp_{y},
\]
the condition $\overline{T_0 f^1} = 0$ is satisfied provided
\[
\mathcal{T}_{q_y} \widetilde{f}^1 = g(q,p_x),
\]
Seeing the function on the right-hand side as a function of $(q_x,p_x)$ indexed by $q_y$ allows to define $\widetilde{f}^1$ pointwise in $q_y$ as 
\begin{equation}
\label{eq:def_f1_tilde}
\widetilde{f}^1 = -(2\pi)^{-N/2}\mathcal{T}_{q_y}^{-1} g.
\end{equation}

Let us now study the regularity of $T_0f^1$. We only treat the term $T_0(f^1 - \widetilde{f}^1)$ since the regularity of $T_0 \widetilde{f}^1$ can be proved similarly. Recall that all the functions under consideration are $C^\infty$ by hypoellipticity. Therefore, only the derivates in the $p$ variables have to be considered because the position space is compact. Now, 
\begin{equation}
\label{eq:expression_f1}
\begin{aligned}
f^1 - \widetilde{f}^1= & -\sum_{i} p_{yi} G'(q_{yi})\mathcal{T}_{q_y}^{-1}(p_{xi}) \\
& -\sum_{i,j,k} p_{yi} G(q_{yj})
\left\{ \left(\mathcal{T}_{q_y}^{-1}\right) \left[
    \partial^2_{q_{yi},q_{xk}} V(q_x,q_y) \partial_{p_{xk}}
    \right]\left(\mathcal{T}_{q_y}^{-1}\right) p_{xj} \right\}.
\end{aligned}
\end{equation}
The $p_y$ dependence is trivial in the above expression, so that only derivatives in~$p_x$ require some attention. Since $T_0 = \A_{y,{\rm ham}} + \mathcal{T}_{q_y}$ where $\A_{y,{\rm ham}} = p_y \cdot \nabla_{q_y} - \nabla_{q_y} V(q_x,q_y) \cdot \nabla_{p_y}$ is an operator in the $q_y,p_y$ variables (parameterized by~$q_x$), it suffices to consider $\mathcal{T}_{q_y} f^1$. This function is, in turn, a linear combination of terms of the form $p_{yi} p_{xi}$ (cf. the first term in the right-hand side of~\eqref{eq:expression_f1}) and $p_{yi} \partial_{p_{xk}}\mathcal{T}_{q_y}^{-1}p_{xj}$ (second term in the right-hand side of~\eqref{eq:expression_f1}). To prove that the latter functions are in~$H^1(\psi_0)$, we use the results of~\cite{Talay02,HN04}, which show that $\mathcal{T}_{q_y}^{-1}$ is a bounded operator on the Hilbert spaces
\[
\left\{ f \in H^m(\Psi_{q_y}) \ \left| \ \int_{(L_x\mathbb{T})^N \times \mathbb{R}^N} f(q_x,p_x) \, \Psi_{q_y}(q_x,p_x) \, dq_x\, dp_x \right. \right \} \subset L^2(\Psi_{q_y})
\]
for any $m \geq 0$, with a bound uniform in~$q_y$.

%--------- gamma_x tend vers l'infini -------------------
\subsection{Proof of Theorem~\ref{prop:gamma_x}}
\label{sec:proof_gamma_x}

The proof follows the same lines as the proof presented in 
Section~\ref{sec:proof_gamma_y}, so we skip the parts of the argument which can
be straightforwardly extended from there.

As in the previous section, we set $m = 1$ to simplify the notation, but the proof can 
be straightforwardly modified to account for more general masses.
Note first that the solution of~\eqref{eq:eq_Poisson_x} is well defined
for any $\gamma_x > 0$, for reasons similar to the ones exposed
at the beginning of Section~\ref{sec:proof_gamma_y}. Define
\[
\mathcal{T}_{q_x} = p_y \cdot \nabla_{q_y} - \nabla_{q_y}V(q_x,q_y) \cdot \nabla_{p_y} 
+ \gamma_y \A_{y, \rm thm},
\]
which is an operator parameterized by $q_x \in (L_x \mathbb{T})^N$, 
and acting on the Hilbert space $L^2(\Psi_{q_x})$, where
\[
\Psi_{q_x}(q_y,p_y) = Z_{q_x}^{-1} \exp\left(-\beta \left( V(q_x,q_y) + \frac{p_y^2}{2m}\right) 
\right).
\]
Its kernel is $\mathrm{Vect}(1) = \{ c \Psi_{q_x}, \ c\in \R\}$.

\paragraph{Formal asymptotic expansion in $\gamma_x$}
We start by formal computations,
with a discussion parallel to the corresponding one in Section~\ref{sec:proof_gamma_y}.
We consider the following \emph{ansatz} for the solution $f_{\gamma_x}$:
\[
f_{\gamma_x} = f^0 + \frac{1}{\gamma_x} f^1 + \frac{1}{\gamma_x^2} f^2 + \dots
\]
and rewrite the operator $\A_0(\gamma_x)$ as the sum $\A_0(\gamma_x) = T_0 + \gamma_x \A_{x, \rm thm}$.
The kernel of the operator $\A_{x, \rm thm}$ on $L^2(\psi_0)$ is
\[
\ker(\A_{x, \rm thm}) = \Big\{ g \in L^2(\psi_0) \ \Big| \ g = g(q,p_y) \Big\}.
\]
Identifying terms with the same powers of $\gamma_x$ 
in~\eqref{eq:eq_Poisson_x}, the following hierarchy is obtained:
\begin{equation}
  \label{eq:hierarchy_x}
  \left \{ \begin{aligned}
    \A_{x, \rm thm}f^0 & = 0, \\
    T_0 f^0 + \A_{x, \rm thm} f^1 & = -\sum_{i=1}^N p_{xi} G(q_{yi}), \\
    T_0 f^1 + \A_{x, \rm thm} f^2 & = 0.
  \end{aligned} \right.
\end{equation}
The first equation shows that $f^0 \equiv f^0(q,p_y)$. The second one can be rewritten
as 
\[
\A_{x, \rm thm} f^1(q,p) = -\sum_{i=1}^N p_{xi}\Big(G(q_{yi})+\partial_{q_{xi}}f^0(q,p_y)\Big)
- \mathcal{T}_{q_x}f^0(q,p_y),
\]
so that 
\[
f^1 = \sum_{i=1}^N p_{xi}\Big(G(q_{yi})+\partial_{q_{xi}}f^0(q,p_y)\Big) + \widetilde{f}^1,
\]
with
\[
\A_{x, \rm thm} \widetilde{f}^1(q,p) = - \mathcal{T}_{q_x}f^0(q,p_y).
\]
The solvability condition requires $\mathcal{T}_{q_x}f^0(q,p_y) = 0$,
hence $f^0 \equiv f^0(q_x)$. Besides, $\widetilde{f}^1$ does not depend on~$p_x$.
The solvability condition for the 
third equation in~\eqref{eq:hierarchy_x} is $T_0 f^1 \in \ker(\A_{x, \rm thm})$.
Now,
\[
\begin{aligned}
T_0\left(\sum_{i=1}^N p_{xi}\Big(G(q_{yi})+\partial_{q_{xi}}f^0(q_x)\Big)\right)
& = -\sum_{i=1}^N \partial_{q_{xi}}V(q_x,q_y) \Big(G(q_{yi})+\partial_{q_{xi}}f^0(q_x)\Big) \\
& \ \ + \sum_{i=1}^N p_{xi}p_{yi} G'(q_{yi}) + \sum_{i=1}^N p_{xi}^2 \partial^2_{q_{xi}}f^0(q_x),
\end{aligned}
\]
and
\[
T_0 \widetilde{f}^1 = \mathcal{T}_{q_x}\widetilde{f}^1 + 
\sum_{i=1}^N p_{xi} \partial_{q_{xi}}\widetilde{f}^1.
\]
We therefore set 
\[
f^0 = 0, 
\qquad
\widetilde{f}^1 = \mathcal{T}_{q_x}^{-1}\left(\sum_{i=1}^N \partial_{q_{xi}}V(q_x,q_y) 
G(q_{yi})\right),
\]
and 
\[
f^2(q,p) = \sum_{i=1}^N p_{xi}\Big( p_{yi}G'(q_{yi}) + 
\partial_{q_{xi}}\widetilde{f}^1\Big) + \widetilde{f}^2,
\]
so that 
\[
\A_{x, \rm thm} f^2(q,p) = -\sum_{i=1}^N p_{xi}\Big( p_{yi}G'(q_{yi}) + 
\partial_{q_{xi}}\widetilde{f}^1\Big) = -T_0 f^1.
\]
The function $\widetilde{f}^2$ is chosen such that $T_0 f^2$ has a vanishing average with respect to the Gaussian measure in the $p_x$ variable. 

Note that $\widetilde{f}^1$ is well defined since $\sum_{i=1}^N \partial_{q_{xi}}V(q_x,q_y) 
G(q_{yi}) \in \ker(\mathcal{T}_{q_x})^\perp = \{ 1 \}^\perp$ (where the orthogonality
is with respect to the scalar product on $L^2(\Psi_{q_x})$). Indeed,
\[
\begin{aligned}
\sum_{i=1}^N \partial_{q_{xi}}V(q_x,q_y) G(q_{yi}) 
& = \sum_{i=1}^N \sum_{j \neq i} \mathcal{V}'(|q_i-q_j|) \frac{q_{xi}-q_{xj}}{|q_i-q_j|} G(q_{yi}) \\
& = \sum_{1 \leq i < j \leq N} \mathcal{V}'(|q_i-q_j|) \frac{q_{xi}-q_{xj}}{|q_i-q_j|} \Big ( 
G(q_{yi}) - G(q_{yj}) \Big).
\end{aligned}
\]
The last line is antisymmetric with respect to the exchange of coordinates $q_{yi}$ and
$q_{yj}$, hence the average of the corresponding function with respect 
to $\Psi_{q_x}$, which is symmetric with respect to the exchange of coordinates $q_{yi}$ and
$q_{yj}$, vanishes.

A discussion similar to the one in Section~\ref{sec:proof_gamma_y} show also that 
$\partial_{q_{xi}}\widetilde{f}^1$ is well defined, hence the definition of $f^2$
makes sense.

\paragraph{Proof of the limit~\eqref{eq:cv_Poisson_gamma_x}}
The remainder of the proof follows the very same lines as the proof presented in 
Section~\ref{sec:proof_gamma_y}, hence we omit it.

%------------------------
%   ACKNOWLEDGEMENTS
%------------------------

\section*{Acknowledgments} 
This work was supported by ANDRA (French National Nuclear Waste Agency) 
and by the French Ministry of 
Research through the grant ANR-09-BLAN-0216-01 (MEGAS).
The authors wish to thank Claude Le Bris, Tony Leli\`evre and 
Greg Pavliotis for fruitful discussions.

\end{document}